\RequirePackage{amsthm}
\documentclass[pdflatex,sn-mathphys-num]{sn-jnl}% Math and Physical Sciences Numbered Reference Style 
\usepackage{multirow}%
\usepackage{amsmath,amssymb,amsfonts}%
\usepackage{mathrsfs}%
\usepackage{xcolor}%
\usepackage{textcomp}%
\usepackage{manyfoot}%
\usepackage{booktabs}%
\usepackage{listings}%
\usepackage{mathtools}
\usepackage{physics}

\theoremstyle{plain}
\newtheorem{theorem}{Theorem}[section]

\newtheorem{proposition}[theorem]{Proposition}
\theoremstyle{definition}
\newtheorem{definition}[theorem]{Definition}
\newtheorem{remark}{Remark}
\newtheorem{example}{Example}

\def\R{\mathbb{R}}
\def\N{\mathbb{N}}
\def\C{\mathbb{C}}
\def\Z{\mathbb{Z}}

\DeclareMathOperator{\End}{End}
\DeclareMathOperator{\Op}{Op}
\DeclareMathOperator{\im}{im}
\renewcommand{\d}[1]{\ensuremath{\operatorname{d}\!{#1}}}

\begin{document}

\title[Article Title]{Generating and Computing Quantum Periods in Exact WKB}


\author{\fnm{Max} \sur{Meynig}, Orcid id: 0000-0002-9015-1923 }\email{max.meynig@uconn.edu} %{0000-0002-9015-1923}
\affil{\orgdiv{Department of Physics}, \orgname{University of Connecticut}, \orgaddress{\street{196A Auditorium Road Unit 3046}, \city{Storrs}, \postcode{06269}, \state{CT}, \country{United States of America}}}

\abstract{
    Periods of rational integrals appear in quantum mechanics through asymptotic expansions of traces computed with the semiclassical symbol calculus.
    In particular, a novel formal series expansion for the trace of the Dirac delta of the quantization of a polynomial Hamiltonian is developed.
    Reducing these integrals modulo total derivatives results in a novel formula for normal forms of the integrals.
    In the case of one degree of freedom, the two dimensional residue formula relates the rational integrals to the quantum actions in the exact WKB formalism.
}




\maketitle
\newpage
\tableofcontents

\section{Introduction}
\label{Introduction}

Motivation for this article comes from exact Wentzel-Kramers-Brillouin (WKB) analysis \cite{vanSpaendonck:2023znn,Alekseev:2024pqg,Colin_de_Verdiere_2005,PhysicsPhysiqueFizika.2.131,AIHPA_1983__39_3_211_0,AIHPA_1999__71_1_1_0,Ashok:2016yxz,Basar_Dunne_Unsal_2017,Codesido_Marino_2018,Delabaere1997ExactSE,Dunham_1932,Zinn-Justin:2004qzw,Zinn-Justin:2004vcw,Raman_Subramanian_2020,Mironov:2009uv,Ito,Bazhanov:2003ni,Bucciotti:2023trp,Gaiotto:2009hg,Cargo_Gracia-Saz_Littlejohn_Reinsch_De_Rios_2004,Konoplya:2019hlu}.
Consider a Schr\"odinger equation of the form
\begin{equation}
    \left(-\frac{\hbar^2}{2}\derivative{{}^2}{x^2} +V(x) \right)\psi(x) = E \psi(x)
\end{equation}
with $V(x)$ a polynomial.
Applying an ansatz $\psi = e^{\frac{i}{\hbar} \int^x \phi \d{x}}$ leads to the equation
\begin{equation}\label{eq::ActionDensity}
    \phi^2 =  2(E - V(x)) + i\hbar \derivative{\phi}{x}.
\end{equation}
which can be solved over formal power series
\begin{equation}\label{eq::WKBseriesZero}
     \phi = \sum_{r=0}^\infty \hbar^r \phi_r
\end{equation}
each term $\phi_r$ is generated recursively from the previous terms.
A reasonable question is the following: \textit{Is it possible to express $\phi_r$ in closed form?}
Here we show that it is possible to find a closed form expression for each term in the series modulo total derivatives.

Each term in the series is a differential on the Riemann surface defined by the classical equation of motion $\frac{1}{2}p^2 +V(x) =E$.
The WKB solutions are resurgent functions with well-characterized Borel singularities \cite{iwaki2014exact,nikolaev2024geometryresurgencewkbsolutions,Delabaere1997ExactSE}.
The closed loop integrals of the differentials in series \eqref{eq::WKBseriesZero} describe the Stokes phenomena of WKB solutions; as a result, they are the fundamental objects appearing in `exact quantization conditions.'
Given boundary conditions, the corresponding exact quantization condition encodes the transseries expansion of the spectrum.
Therefore, a fundamental task in exact WKB analysis of both solutions and spectra is to evaluate the closed loop integrals
\begin{equation}
    \oint_\gamma \phi \,\d x = \sum_{r=0}^\infty \hbar^r \oint_\gamma \phi_r \,\d x
\end{equation}
along the various cycles $\gamma$ of the underlying Riemann surface.

This task can be simplified.
From Stokes' theorem, total derivatives integrate to zero.
It is therefore advantageous to consider the de Rham cohomology of the underlying Riemann surface, as it describes the classes of differentials which are equal modulo total-derivatives.
Because the first de Rham cohomology group of the Riemann surface is finitely generated, with rank twice the genus of the Riemann surface, it is only necessary to calculate a finite number of integrals.
The remainder of the work comes from strategically applying integration-by-parts until each term is expressed as a linear combination of the chosen generators. %for the first de Rham cohomology group.

In other words, because the first de Rham cohomology group is a finite dimensional vector space, it has a basis $\{ \alpha_i \}$ and there are coefficients $\xi_{r,i}(E)$ such that
\begin{equation}
    \oint \phi_r \d x = \sum_{i=1}^{ 2\times(\text{genus}) } \xi_{r,i}(E) \oint \alpha_i.
\end{equation}
Because the $\phi_r$ are generated by a recursive formula, a reasonable question is the following: \textit{Do the $\xi_{r,i}(E)$ satisfy a recursive equation? And if the $\phi_r$ can be found in closed form can the $\xi_{r,i}(E)$?}

To answer the above questions, we adopt a new perspective.
In \cite{Colin_de_Verdiere_2005} Colin-de-Verdi{\`e}re proved that up to total derivatives, the $\phi_r$ can be generated using the semiclassical symbol calculus (similar approaches were also developed in \cite{Cargo_Gracia-Saz_Littlejohn_Reinsch_De_Rios_2004,PhysicsPhysiqueFizika.2.131}).
By applying the multivariable residue calculus to the expressions coming from the symbol calculus, the first question is addressed.
As a corollary, we find a novel expression for $\derivative{\phi}{E}$ up to total derivatives.

Using this expression, the second question is addressed.
A novel method is presented which computes the decomposition of $\derivative{\phi}{E}$ into a basis for the de Rham cohomology group.
It completes the necessary integration-by-parts steps needed to decompose \textit{all} terms of $\derivative{\phi}{E}$ into a basis, in a finite number of steps.

As pointed out by Colin-de-Verdi{\`e}re \cite{Colin_de_Verdiere_2005}, the perspective supplied by the symbol calculus generalizes to higher dimensions straightforwardly.
Additionally, the method we give for the decomposition of the integrals is also well suited to multidimensional integrals.
For these reasons, the perspective adopted here is supplied by the symbol calculus and the principal object of study is $\tr \delta (\hat H - E)$, where $\delta$ is the dirac delta and $\hat H $ is the differential operator arising from the quantization of a polynomial.

Resurgence implies that non-perturbative structure is encoded in the asymptotics of the coefficients of factorially divergent series \cite{Aniceto:2018bis}.
The framework introduced here provides an efficient algorithm for calculating the expansion of the spectral delta function $\tr \delta (\hat H - E)$, making high orders accessible and thereby providing an interface with resurgence.
This is particularly interesting in higher dimensional systems where the non-perturbative structure of the spectral delta function---the leading terms of which are provided by Gutzwiller's trace formula \cite{Gutzwiller_1971}---can be difficult to compute in practice.

\section{Description of main results}
\label{sec::MainResults}

The semiclassical symbol calculus provides a pathway to formal series expansions for traces of `quantum mechanical' operators.
The standard quantization, as well as `other' quantization schemes are defined according to the following \cite{Zworski}:

\begin{definition}\label{def::QuantizationFormulae}
    Let $a:\R^{2n}  \to \C$ be a function on the phase space.
    The function $a(x,p)$ is referred to as a \textit{symbol}.
    Let $0\leq t \leq 1$.
    The linear differential operator $\text{Op}_t(a) $ given by
    \begin{equation*}
        \Op_t(a) u = \frac{1}{(2\pi \hbar )^n } \int_{\R_n} \int_{\R_n} e^{\frac{i}{\hbar}(x - y)\cdot p } a\left(t x + (1-t) y, p \right) u(y) \d y \d p.
    \end{equation*}
    is called a \textit{semiclassical pseudodifferential operator}, or the ``quantization of $a$.''
    The case $t=1$ defines the standard quantization.
\end{definition}

\begin{example}
    Let $n=1$ and consider the standard quantization (the case $t=1$) of $ x p $.
    From the definition and properties of the fourier transform
    \begin{equation}
        \begin{aligned}
            \Op_1(x p) u  &=  \frac{1}{2\pi \hbar  } \int_{\R} \int_{\R} e^{\frac{i}{\hbar}(x - y)\cdot p }  x p u(y) \d y \d p
            \\ &= \frac{x}{2\pi \hbar  } \int_{\R} \int_{\R} e^{\frac{i}{\hbar}(x - y)\cdot p }  p u(y) \d y \d p
            \\ &= x \frac{\hbar}{i} \frac{\partial}{\partial x} u(x).
        \end{aligned}
    \end{equation} 
    In addition, consider the Weyl quantization (the case $t=1/2$) of $ x p $.
    A similar calculation shows that 
    \begin{equation}
        \begin{aligned}
            \Op_{\frac{1}{2}} (x p) u  &=  \frac{1}{2\pi \hbar  } \int_{\R} \int_{\R} e^{\frac{i}{\hbar}(x - y)\cdot p } \left(\frac{1}{2}x + \frac{1}{2}y\right) p u(y) \d y \d p
            \\ &= \frac{1}{2} \left( x \frac{\hbar}{i} \frac{\partial}{\partial x} +  \frac{\hbar}{i} \frac{\partial}{\partial x} x \right) u(x).
        \end{aligned}
    \end{equation} 
\end{example}

The expression for $\tr \delta (\hat H - E)$ features a sequence of polynomials in non-commutative differential operators.
The polynomials are closely related to the solution to noncommutative recurrence relations with constant coefficients.
The following definition is inspired by notation used in \cite{Jivulescu_2007,Puhlfurst:2015foi} and is helpful for expressing these polynomials.
\begin{definition}\label{def::SumOfNonCommutingPermutations}
    Let $t$ be an integer greater than zero.
    For $a_1,a_2,\dots,a_t \in \N_0$ and for $g_1,g_2,\dots,g_t$ elements of a ring let $\{g_1^{(a_1)}g_2^{(a_2)}\dots g_t^{(a_t)} \}$
     denote the sum of all possible permutations of $a_1$ factors of $g_1$, $a_2$ factors of $g_2,\dots $ and $a_t$ factors of $g_t$.
\end{definition}

In terms of these definitions the first result is:
\begin{theorem}\label{thm::MainResultOne}
    Let $E \in \C$.
    Let $H:\C^{2n} \to \C $ be a polynomial and let $\hat H =\Op_1 (H)$ be the standard quantization of $H$.
    Then there is a cycle $\Gamma$ such that $ \tr \delta(\hat H - E) $ has a formal series expansion in $\hbar$ given by
    \begin{equation}\label{eq::MainTheoremSeries}
        (2\pi \hbar )^n \tr \delta(\hat H - E) =\sum_{k=0}^{\infty} \frac{\hbar^k}{2\pi i} \int_\Gamma P_k (g_1,\dots,g_{\deg H}) \frac{\d x_1  \d p_1 \dots  \d x_{n}  \d p_{n}}{H(x,p) - E}
    \end{equation}
    where
    \begin{equation}\label{eq::MainTheoremGs}
        g_r = -\frac{1}{H-E} \frac{1}{i^{r} r!} \sum_{\beta_1 +\dots +\beta_n = r }\frac{\partial^{r} H }{\partial p_1^{\beta_1} \dots \partial p_n^{\beta_n}}  \frac{\partial^{r}  }{\partial x_1^{\beta_1} \dots \partial x_n^{\beta_n}}
    \end{equation}
    for each $r = 1,\dots,\deg H$ and
    \begin{equation}\label{eq::MainTheoremPs}
        P_k = \sum_{a_1 +2 a_2 + \dots + \deg H a_{\deg H} = k} \{ g_1^{(a_1)}\dots g_{\deg H}^{(a_{\deg H})} \}
    \end{equation}
    for each $k$ in the sum.
\end{theorem}

A proof of theorem \ref{thm::MainResultOne} is given in section \ref{sec::ProofOfTheoremOne}.

To clarify both definition \ref{def::SumOfNonCommutingPermutations} and theorem \ref{thm::MainResultOne} consider the following example, which describes the special case of a one dimensional Schr\"odinger operator. 
\begin{example}
    Suppose $n=1$ and $H = p^2/2 +V(x)$ then the only non-zero operators $g_1, g_2, \dots ,g_{\deg H}$ are given by
    \begin{equation}
        \begin{aligned}
            & g_1 = \frac{i p}{H-E}   \frac{\partial  }{\partial x}, &
             && g_2 = \frac{1}{2}\frac{1}{H-E}   \frac{\partial^2  }{\partial x^2}.
        \end{aligned}
    \end{equation}
    And the first four of the $P_k$ are given by
    \begin{equation}
        \begin{aligned}
            P_1 &= g_1,
        \\  P_2 &= g_1^2 + g_2,
        \\  P_3 &= g_1^3 + g_1 g_2 + g_2 g_1,
        \\  P_4 &= g_1^4 + g_1^2 g_2 +g_1 g_2 g_1 + g_2 g_1^2 +g_2^2.
        \end{aligned}
    \end{equation}

\end{example}

\subsection{Restating the reduction problem}

Because the integrals in theorem \ref{thm::MainResultOne} are `rational integrals,' (integrals of rational functions) they can be simplified using the Griffiths-Dwork algorithm \cite{GriffithsI,GriffithsII,Dwork1962,Dwork1964,Lairez_2015,Morrison:1991cd,Doran:2023yzu,lairezAlgorithmsMinimalPicard2023,delaCruz:2024xit,Muller-Stach_Weinzierl_Zayadeh_2014}. % and, if necessary, the extension due to Dimca, Saito and Lairez \cite{Lairez_2015,dimca2006generalizationgriffithstheoremrational}.
Given a rational differential form, the Griffiths-Dwork reduction outputs a new rational differential which is a normal form of the input modulo `total derivatives' \cite{Lairez_2015}.
Therefore, it can be applied to the task at hand in a term-by-term manner, first by expanding the above expression and then reducing the terms to normal form.

The rational differentials considered here should be thought of as integrals on \textit{the complex projective space} $\mathbb{P}^{2n}$.
Given a complex projective hypersurface $X\subset \mathbb{P}^{2n}$ of codimension one, assumed to be irreducible and smooth (in the current context $X$ is determined by $H=E$), 
the Griffiths-Dwork reduction \cite{Lairez_2015} provides a linear map $r_{\text{GD}}$ from the vector space $A^{2n}(X)$ of rational differential $2n$-forms on $\mathbb{P}^{2n} \setminus X$ to itself.
That is, $r_{\text{GD}}: A^{2n}(X) \to A^{2n}(X)$.
The filtration of $A^{2n}(X)$ by \textit{pole order} 
$$ A^{2n}_1(X) \subseteq A^{2n}_2(X)\subseteq \dots \subseteq  A^{2n}_{m-1} (X) \subseteq A^{2n}_{m} (X) \subseteq \dots \subseteq A^{2n}(X) $$
determines an essential property of $r_{\text{GD}}$.
The map $r_{\text{GD}} $ satisfies 
$$r_{\text{GD}}(A^{2n}_{m} (X)) \subseteq A^{2n}_{m-1} (X)$$
 for all $m\geq ((\deg H - 1 )(2n+1)+1)/ \deg H $, see \cite{GriffithsI} proposition 4.13.
Additionally, for each $a \in  A^{2n}(X)$ and any $2n$-cycle $\Gamma$ of $\mathbb{P}^{2n}\setminus X$ the map satisfies
\begin{equation}
    \int_{\Gamma} a  = \int_{\Gamma}r_{\text{GD}}(a).
\end{equation}
That is, $r_{\text{GD}}$ `lowers pole orders' by subtracting total derivatives.
Iterated application of $r_{\text{GD}}$ yields fixed points which are the `normal forms' described above.

In general, reducing an $a \in A^{2n}_m(X)$ to its normal form requires at most $m$ iterations and at minimum 
$$ l_{\min} = \max(0,m- \lfloor ((\deg H - 1 )(2n+1)+1)/ \deg H \rfloor).$$
The problem is as follows:
Each operator $g_k \in \{g_1,\dots,g_{\deg H}\}$ raises the pole order by $k+1$.
That is, 
$$ g_k (A^{2n}_{m}(X)) \subseteq  A^{2n}_{m+k+1}(X) .$$
Therefore with $P_k$ as in theorem \ref{thm::MainResultOne} we have that, 
$$ P_k(g_1,\dots,g_{\deg H})(A^{2n}_{m}(X)) \subseteq A^{2n}_{m+k\, (k+1)}(X) $$
in other words, computing normal forms, for each term in the series in \ref{thm::MainResultOne}, requires an ever growing number of iterations of $r_{\text{GD}}$, with maximum $m+k\, (k+1)$ needed at the $k$th order.

\begin{remark}
    We propose a new method which limits the number of iterations needed to reduce the $k$-th order of series \eqref{eq::MainTheoremSeries} to a constant. 
    That is, we eliminate the `$k$ dependence' described above.
\end{remark}

\subsection{Main result on the reduction of the integrals}

The modified reduction procedure applies directly to the differential operators $g_1,\dots,g_{\deg H}$ defined in theorem \ref{thm::MainResultOne}.
The procedure reduces `all' integrals by computing a ring homomorphism from the ring $G$ generated by addition and composition of the $g_1,\dots,g_{\deg H}$, to the following ring of matrices.

\begin{definition}\label{def::RingOfMatrixFamilies}
    Let  $\mathcal{M}_{d\times d}(\C[t])$ denote the set of $d$-by-$d$ square matrices with coefficients in $\C[t]$.
    Let $\mathcal{R}_{n,d}[t]$ be the set of matrix valued functions
    $$g^{(\cdot, \cdot)}: \mathbb{N}_0^{2n+1}\times \mathbb{N}_0^{2n+1} \to \mathcal{M}_{d\times d}(\C[t])$$
    such that $g^{\alpha , \cdot }:\mathbb{N}_0^{2n+1} \to  \mathcal{M}_{d\times d}(\C[t]) $ has finite support for all $\alpha \in \N^{2n+1}_0$.
    The set $R_{n,d}[t]$ is given the structure of a ring $(\mathcal{R}_{n,d}[t],+ ,\diamond)$ by the product $\diamond$ defined as follows:
    Let $g_1,g_2 \in \mathcal{R}_{n,d}[t]$ and let $I_{\alpha} \subset \N^{2n+1}_0$ denote the support of $g_1^{\alpha , \cdot} : \N_0^{2n+1} \to \mathcal{M}_{d\times d}(\C[t])$.
    Then
    \begin{equation*}
        g_1 \diamond  g_2 = \sum_{\beta \in I_{(\cdot)} } g_1^{\cdot, \beta}g_2^{\beta ,\cdot}
    \end{equation*}
\end{definition}

\begin{remark}
    The multiplicative identity in $\mathcal{R}_{n,d}[t]$ is given by the function $\delta^{\alpha , \beta}$ which is the identity matrix if $\alpha = \beta$ and the zero matrix otherwise.
\end{remark}

The ring homomorphism is accompanied by a linear map which describes an action of the elements of $\mathcal{R}_{n,d}[t]$ on rational differential forms.
The key properties of these two maps are described by the following theorem, which we prove constructively in section \ref{sec::ProofOfTheoremTwo}.
\begin{theorem}\label{thm::MainResultTwo}
    Let $H,E,g_1,\dots,g_{\deg H}$ be as in theorem \ref{thm::MainResultOne} and suppose the complex projective variety $X$ determined by $H=E$, is irreducible and smooth.
    Let $r = \lfloor ((\deg H - 1 )(2n+1)+1)/ \deg H \rfloor$.
    Then there is an integer $d$ such that there is a ring homomorphism $\lambda  : G \to \mathcal{R}_{n,d}[t]$ and a linear map
    $\mu : \im \lambda \otimes A^{2n}(X)  \to A^{2n}(X)$,
    such that 
    $$\mu( \im \lambda \otimes A^{2n}_{m}(X))  \subseteq A^{2n}_{m+r}(X)$$
    and for all $g \in G$ and $a\in A^{2n}(X)$ we have that 
    $$ \int_\Gamma ga = \int_\Gamma \mu(\lambda(g)\otimes a). $$
\end{theorem} 
\begin{remark}
    Observe that theorem \ref{thm::MainResultTwo} states that the `$k$' dependence in the pole order of $P_k (a)$ can be removed.
\end{remark}

To summarize, the method computes `all' normal forms in a finite number of steps.
Once the images $\lambda(g_i)$ have been found, the work of reducing $P_k(g_1,\dots,g_{\deg H})(a)$ to a normal form amounts to addition and multiplication of matrices.
Moreover, the proof is constructive and the images $\lambda(g_i)$ are determined by finite computations.

The construction of the homomorphism $\lambda : G\to \mathcal{R}_{n, d}(\C[t])$ is given in section \ref{sec::CompleteReduction} and is summarized as follows:
First, the rational differential operators $g_i$ of theorem \ref{thm::MainResultOne} are mapped to elements $\hat g_i$ of a Weyl algebra which acts on the space of numerators of the rational differentials, this is accomplished by proposition \ref{prop::QuantumOperatorsMapped}.
Bringing the operators $\hat{g}_i$ into a normal form is equivalent to computing their projections modulo the right ideal $J_t = \langle \pdv{f}{z_0} - t \pdv{}{z_{0}}, \dots, \pdv{f}{z_{2n}}- t \pdv{}{z_{2n}} \rangle $ of the Weyl algebra.
This task is completed through the introduction of a reduction method which avoids non-commutative Gr\"obner basis computations and is designed to mimic the Griffiths-Dwork reduction---described below in section \ref{sec::TheGriffithsDworkreduction}---as closely as possible.
The method is sketched below:
\begin{enumerate}
    \item Construct a monomial basis $\{e_j\}_{j=1}^d$ for the quotient ring $\C[z]/J_0$ with $J_0$ the ideal of $\C[z]$ given by $J_t$ with $t$ set to zero.
    \item Compute the normal ordering of $e_j \partial^\alpha \hat g_i$ for arbitrary $\alpha \in \N_0^{2n+1}$
    \item Iteratively, remove (left) polynomial factors of $\pdv{f}{z_{i}}$ from this normal ordering, replacing them with the corresponding derivatives $t\pdv{}{z_{i}}$.
    \item The resulting expression includes only polynomial factors in the basis $\{e_j\}_{j=1}^d$ and is therefore of the form $\sum_{k,\beta} g_{i,j,k}^{\alpha,\beta}(t) e_k \partial^\beta$.
    The coefficients $g_{i,j,k}^{\alpha,\beta}(t)$ are the coefficients of the matrices determined by $\lambda(g_i)(\alpha, \beta)$.
\end{enumerate}
The map $\mu$ will be described in detail in section \ref{sec::CompleteReduction} and \ref{sec::ProofOfTheoremTwo}.
A subtlety comes from the introduction of the parameter $t$.
For the correct definition of $\mu$, the parameter $t$ must be replaced by the inverse of an element $\hat m$ of the Weyl algebra---introduced in proposition \ref{prop::QuantumOperatorsMapped} to count pole orders---which is a shift and rescaling of the Euler vector field. %$ = \frac{1}{\deg f} \left(\sum_{i=0}^{2n} z_i \pdv{}{z_i} + 2n +1 \right)$ 
Due to the non-commutative nature of this element, the map is given explicitly by $t\mapsto S_{\hat m} \frac{1}{\hat m}$ where $S_{\hat m}$ is a right shift operator defined by $q(\hat m) S_{\hat m} = S_{\hat m} q(\hat m +1)$ for any polynomial $q$.

\subsection{Notation}
\label{sec::Notation}
The multi-index notation will be used extensively.
For a $\beta \in \N^{n+1}_0$ given by $\beta = (\beta_0 , \dots,\beta_n)$ the notation $|\beta|$ is defined by
\begin{equation}
    | \beta | = \beta_0 + \beta_1 + \dots + \beta_n.
\end{equation}
Given $z = (z_0,\dots,z_n)$ the expression $z^\beta$ is defined by
\begin{equation}
    z^{\beta} = z_0^{\beta_0} z_2^{\beta_2} \dots z_{n}^{\beta_{n}}
\end{equation}
as is typical in the multi-index notation.
Additionally, given a multi-index $\alpha = (\alpha_0 , \dots ,\alpha_{n})$ the multi-index factorial $\alpha!$ is defined as
\begin{equation}
    \alpha! = \alpha_0 !\alpha_1 !  \dots \alpha_{n}!
\end{equation}

The following notation will be used to denote ideals of both commutative and noncommutative rings.
Given a ring $R$ and a collection of subsets $r_1 ,r_2 , \dots, r_k \subset R$ then 
\begin{equation}
    \langle r_1,\dots,r_k \rangle
\end{equation}
denotes the smallest ideal of $R$ which contains $r_1,r_2,\dots , r_k$.
In the case $R$ is non-commutative the notation indicates the smallest right ideal of $R$ which contains $r_1,r_2,\dots , r_k$.

\section{Reduction of the integrals}
\label{sec::ReductionOfTheIntegrals}

We review the Griffiths-Dwork reduction for smooth hypersurfaces in section \ref{sec::TheGriffithsDworkreduction}.
This provides sufficient background for section \ref{sec::CompleteReduction}, where the necessary results used in the proof of theorem  \ref{thm::MainResultTwo} are developed.
In section \ref{sec::ProofOfTheoremTwo} theorem \ref{thm::MainResultTwo} is proven.

\subsection{The Griffiths-Dwork reduction}
\label{sec::TheGriffithsDworkreduction}

As Griffiths points out \cite{GriffithsI}, rational one dimensional residue integrals are best considered as integrals of differentials on the Riemann sphere $\mathbb{P}^{1} = \C\cup \{\infty \}$.
Similarly, the rational $2n$-forms defined by theorem \ref{thm::MainResultOne} (and more generally the differential forms which are treated by the Griffiths-Dwork reduction) are best thought of as rational differentials on the complex projective space $\mathbb{P}^{2n}$.
While the interpretation of $\mathbb{P}^{2n}$ as $\C^{2n}$ with the hyperplane at infinity added provides geometric `intuition,' the equivalent description of $\mathbb{P}^{2n}$ as the space of lines through the origin of $\C^{2n+1}$ is more convenient for calculations.
The latter is described succinctly with homogeneous coordinates.
In line with these observations, the first step in describing the Griffiths-Dwork Reduction is to introduce homogeneous coordinates.

Given local coordinates $(x,p)$ on phase space where|the $x$ are the `position' coordinates and $p$ the momentum coordinates|define a set of homogeneous coordinates $[ z_0 ,\dots, z_{2n}]\in \mathbb P^{2n}$ on the complex projective space according to
\begin{equation}\label{eq::homogeneousCoordinates}
    \begin{aligned}
        &x_{j} = \frac{z_{2j-1}}{z_0} & \text{ and } && p_j = \frac{ z_{2j} }{z_0}.
    \end{aligned}
\end{equation}
In the geometric perspective, the equation $z_0 = 0$ describes the hyperplane at infinity.
The following polynomial
\begin{equation}\label{eq::fPolynomial}
    f = z_0^{\deg H}( H(z_{2i-1}/z_0 ,z_{2i}/z_0 ) - E)
\end{equation}
will play a central role in the analysis that follows.
The algebraic variety defined by the equation $f=0$ will be denoted $X$.
Exchanging the local coordinates for homogeneous coordinates in the integrals in theorem \ref{thm::MainResultOne} gives
\begin{equation}\label{eq::ClassicalPeriodInHC}
    \frac{1 }{H - E} \d x_1 \wedge \d p_1 \wedge \dots \wedge \d x_{n} \wedge \d p_n = z_0^{\deg f - 2n - 1}\frac{\Omega}{f}.
\end{equation}
where $ \Omega =  \sum_{i=0}^{2n} (-1)^i z_i \, \d z_0 \wedge \dots \wedge \widehat{\d z_i} \wedge \dots \wedge \d z_{2n}$ and the hat indicates omission.

Proving theorem \ref{thm::MainResultTwo} requires Griffiths' \cite{GriffithsI} results on the structure of $A^{2n}_m(X)$ as well the corresponding de Rham cohomology group. 
The advantage of working with homogeneous coordinates manifests in describing these spaces of forms.
The following definitions lead to a definition of the rational de Rham cohomology group as well as a theorem of Griffiths on its structure.
The map $r_{\text{GD}}$ described in section \ref{sec::MainResults} is also made concrete.

The following definition enumerates a finite set of polynomials.
The definition is used to describe $A^{2n}_m(X)$, the rational $2n$-forms on $\mathbb{P}^{2n} \setminus X$ with pole order $m$.
\begin{definition}\label{def::Fsets}
    Let $F_m$ be given by
    $$F_m \coloneqq \text{span} \{ z^\alpha \mid |\alpha | = m \deg f - 2n -1 \}$$
    if  $m \geq \frac{2n+1}{\deg f}$, and $F_m \coloneqq \{0\}$ if $m< \frac{2n+1}{\deg f}$.
    Additionally, define $F \coloneqq \cup_m F_m$.
\end{definition}

In \cite{GriffithsI}, Griffiths proved that the vector space $A^{2n}_{m}(X)$ of $2n$-forms on $\mathbb{P}^{2n} \setminus X$ with pole of order $m$ is given by
\begin{equation}
    A^{2n}_{m}(X) = \left\{ \frac{q_m}{f^m} \Omega \mid q_m \in F_m \right\}
\end{equation}
and that the elements of the vector space $A^{2n-1}_{m-1}(X)$ of $(2n-1)$-forms on $\mathbb{P}^{2n}\setminus X$, with pole of order $m-1$, is comprised of elements of the form
\begin{equation}\label{eq::exactForm}
    \sum_{i<j} \frac{z_i b_i - z_j b_i }{f^{m-1}} \d z_0 \wedge \dots \wedge \widehat{\d z_i} \wedge \dots \wedge \widehat{\d z_j}\wedge \dots \wedge \d z_{2n}
\end{equation}
where the $b_i$ are any homogeneous polynomials of degree $(m \deg f -2n)$ for each $i=0,\dots,2n$.
Following Griffiths, \cite{GriffithsI} these results can be used to define the rational de Rham cohomology group.

\begin{definition}\label{def::RationaldeRhamGroup}
    For $m \in \N $ the quotient $\mathscr{H}^{m}(X) = A^{2n}_m(X)/\d A^{2n-1}_{m-1}(X)$ defines the rational de Rham cohomology group of pole order $m$.
    The \textit{rational de Rham cohomology group} $\mathscr{H}(X)$ is defined through $$\mathscr{H}(X) = \lim_{m\to \infty} \mathscr{H}^{m}(X).$$
\end{definition}

Note that $dA^{2n-1}_{m-1}(X) \subset A^{2n}_m(X)$ is comprised of elements of the form
\begin{equation}
    \sum_{i=0}^{2n} \frac{ f \frac{\partial a_i }{\partial z_i}  -  (m-1) a_i \frac{\partial f}{\partial z_i} }{f^m} \Omega
\end{equation}
where $a_i$ is homogeneous of degree $\deg a_i = (m-1) \deg f - 2(n +1)$ for each $i$.
This observation forms the basis of the Griffiths-Dwork reduction.
Terms of higher pole order can be replaced with terms of lower pole order if they lie in the Jacobian ideal
\begin{equation}\label{eq::JacobianIdeal}
    J  = \langle \partial f/\partial z_0, \dots, \partial f/\partial z_{2n} \rangle .
\end{equation}
For the case $X$ smooth|that is, where $\frac{\partial f}{\partial z_0}=0,\dots,\frac{\partial f}{\partial z_{2n}}=0$ have no common solutions in $\mathbb{P}^{2n}$|Griffiths proved the converse.
Meaning, if a form can be replaced with one of lower pole order, then the polynomial in the numerator is in $J$.
When the assumption fails, the converse is not true and additional relations must be taken into account \cite{dimca2006generalizationgriffithstheoremrational,Lairez_2015}.

The above discussion describes a result of Griffiths \cite{GriffithsI} on the structure of $\mathscr{H}(X)$ which can be summarized as follows.
\begin{theorem}[Griffiths \cite{GriffithsI}]\label{thm::Griffiths}
    Suppose $f=0$ is smooth.
    Let $J$ be the ideal $J = \langle \partial f/\partial z_0, \dots, \partial f/\partial z_{2n} \rangle$ of $\C[z]$ and let $F$ be given by definition \ref{def::Fsets}.
    Then there is a vector space isomorphism $\mathscr{H}(X) \to F /J $.
\end{theorem}

\begin{remark}
    The singular case has been treated by Dimca, Saito \cite{dimca2006generalizationgriffithstheoremrational} and Lairez \cite{Lairez_2015}.
    If the hypersurface $f=0$ is singular there is a hierarchy of new relations, the first of which comes from syzygies of the ideal $J$.
\end{remark}

The map $r_{\text{GD}}$ described in section \ref{sec::MainResults} can now be described more precisely.
Given a differential form $ \frac{q_m \Omega}{f^m} \in A^{2n}_m(X)$ there are polynomials $a_0,\dots,a_{2n}$ and $r$ such that
\begin{equation}
    q_m = \sum_{i=0}^{2n} a_i \frac{\partial f}{\partial z_i } +r
\end{equation}
and the map $r_{\text{GD}}$ is given by
\begin{equation}
    r_{\text{GD}}\left( \frac{q_m \Omega}{f^m} \right) = \frac{1}{m} \sum_{i=0}^{2n} \frac{\partial a_i}{\partial z_i } \frac{\Omega}{f^{m-1}} + \frac{r \Omega}{f^m}.
\end{equation}
\begin{remark}
    For practical computations, completing the decomposition of $q_m$ into the Jacobian ideal requires the use of a Gr{\"o}bner basis for the ideal $J$ as well as the choice of a monomial order.    
\end{remark}

\begin{example}\label{ex::PicardFuchsEquation}
    An application of the Griffiths-Dwork reduction is to the computation of Picard-Fuchs equations \cite{Lairez_2015}.
    For example, suppose $H(x,p) =\frac{1}{2} p^{2} +x^{2} + x^{3}$.
    The homogeneous polynomial $f$ defined in equation \eqref{eq::fPolynomial} is given by 
    \begin{equation}\label{eq::cubicf}
        f = \frac{1}{2} z_{2}^{2} z_{0} +z_{1}^{2} z_{0}+z_{1}^{3} -E z_{0}^{3}. 
    \end{equation}
    Consider the differential form $\frac{\Omega}{f}$ and its first two derivatives
    \begin{equation}
        \begin{aligned}
            &\derivative{}{E} \frac{\Omega}{f} = \frac{z_0^3 \Omega}{f^2}
        & \text{ and, }  &&\derivative{^{2}}{E^2} \frac{\Omega}{f^2} = 2\frac{z_0^6 \Omega}{f^3} .
        \end{aligned}
    \end{equation}
    We have that 
    \begin{equation}
        \begin{aligned}
            r_\text{GD}\left(\derivative{}{E} \frac{\Omega}{f}\right) &=   \frac{9}{2(4- 27 E)} \frac{\Omega}{f} -\frac{1}{E (27 E-4)}z_1 z_2^2\frac{\Omega}{f^2}   ,
            \\(r_\text{GD}\circ r_\text{GD})\left(\derivative{^{2}}{E^2} \frac{\Omega}{f}\right) & = \frac{3 (189 E-4)}{4 (4-27 E)^2 E}\frac{\Omega}{f}  + \frac{2(27 E-2) }{(4-27 E)^2 E^2} z_1 z_2^2\frac{\Omega}{f^2}
        \end{aligned}
    \end{equation}
    which shows, up to exact forms, that $\derivative{^{2}}{E^2} \frac{\Omega}{f}$ is a linear combination of $\derivative{}{E} \frac{\Omega}{f}$ and  $\frac{\Omega}{f}$.
    That is, there is an $\eta \in A^{2n-1}(X)$ such that 
    \begin{equation}
        4 \left(27 E^{2}-4 E\right) \derivative{^{2}}{E^2} \frac{\Omega}{f} + 8 \left(27 E-2\right) \derivative{}{E} \frac{\Omega}{f}  + 15 \frac{\Omega}{f} =d\eta.
    \end{equation}
    Integrating gives a differential equation for the associated period integrals $ \int_{\gamma} \frac{\Omega}{f}$.
    That is,
    \begin{equation}
        4 \left(27 E^{2}-4 E\right) \derivative{^{2}}{E^2} \int_{\gamma}  \frac{\Omega}{f} + 8 \left(27 E-2\right) \derivative{}{E}  \int_{\gamma}  \frac{\Omega}{f}  + 15 \int_{\gamma} \frac{\Omega}{f} =0.
    \end{equation}
    This is the Picard-Fuchs equation.
\end{example}

\subsection{Complete reduction of the integrals}
\label{sec::CompleteReduction}

Theorem \ref{thm::Griffiths} and the preceding discussion contain the essential ingredients of the Griffiths-Dwork reduction for smooth hypersurfaces.
The method discussed in this paper offers a `reduction' of the operators $g_i$ from theorem \ref{thm::MainResultOne}. %$ \in \End(A^{2n}(X))$.
A simplifying step is to work directly with the polynomials $F$ where the $g_i$ can be interpreted as endomorphisms of polynomials.
To this end consider the following map from $F$ to $A^{2n}(X)$.
\begin{definition}\label{def::PsiMap}
    Define a map $\psi : F \to A^{2n}(X)$ by the following:
    For $q\in F$ let $q_m \in F_m$ be such that $q = \sum_m q_m$ then
    \begin{equation*}
        \psi (q) = \sum_m \frac{q_m}{f^m}\Omega.
    \end{equation*}
\end{definition}
In particular we have that
\begin{proposition}\label{prop::QuantumOperatorsMapped}
    Let $g_1,\dots,g_{\deg H}$ be as in proposition \eqref{thm::MainResultOne} and let $\hat m = \frac{1}{\deg f} \left(\sum_{i=0}^{2n} z_i \partial/\partial z_i + 2n +1 \right)$.
    For $k=1,\dots ,\deg H$ let
    \begin{equation*}
         \hat g_k = \frac{z_0^k}{i^k k!} \sum_{|\beta| = k} \frac{\partial^k f}{\partial z_{2}^{\beta_1}\partial z_{4}^{\beta_2}\dots\partial z_{2n}^{\beta_{n}}} \prod_{i=1}^{n} \left( z_0 \frac{\partial}{\partial z_{2 i -1}} -z_0 \frac{\partial f}{\partial z_{2 i -1}} \hat m  \right)^{\beta_i}.
    \end{equation*}
    Then $$ \psi( \hat g_k q) = g_k \psi(q) .$$
\end{proposition}
\begin{proof}
    Let $m_1,m_2 \in \N$.
    Let $q_{m_1}\in F_{m_1}$ and let $q$ be homogeneous of degree $ m_2 \deg f$. 
    Then $q q_{m_1}\in F_{m_2+m_1}$.
    The multiplicative part of the operator follows from
    \begin{equation}
        \psi (q q_{m_1} ) = \frac{q q_{m_1} \Omega}{f^{m_2+m_1}}= \frac{q }{f^{m_2}}\psi ( q_{m_1} ).
    \end{equation}
    To address the rest of the operator first note the following.
    Let $m \in \N$ and let $q_{m} \in F_{m}$.
    Euler's homogenous function theorem implies that $\hat m (q_{m}) = {m} q_{m}$.
    Let
    $$\delta_i  =z_0 \left(  \frac{\partial}{\partial z_{2 i -1}} - \frac{\partial f}{\partial z_{2 i -1}} \hat m  \right).$$
    It remains to check that $\psi( \delta_i q_m ) =  \frac{\partial}{\partial x_{i}} \psi (q_m).$
   Observe
    \begin{equation}
        \begin{aligned}
            \psi(\delta_i q_m) &= \psi \left(z_0 \left(  \frac{\partial q_m }{\partial z_{2 i -1}} - \frac{\partial f}{\partial z_{2 i -1}} m q_m \right) \right)
            \\ & =  z_0 \frac{\partial q_m }{\partial z_{2 i -1}}\frac{\Omega}{f^m} - m \frac{\partial f}{\partial z_{2 i -1}} \frac{q_m \Omega }{f^{m+1}}
            \\ & = \left(\frac{\partial q_m }{\partial x_{i}}\frac{\omega}{(H-E)^m} - m \frac{\partial H}{\partial x_{i }} \frac{q_m \omega }{(H-E)^{m+1}} \right)\lvert_{z_0 = 1}
            \\ & =\frac{\partial}{\partial x_i}\frac{ q_m\lvert_{z_0 = 1} }{(H-E)^m}\omega.
        \end{aligned}
    \end{equation}
    Therefore the operators $\hat g_k $ reduce to the operators $g_k$ under the map $\psi$ and then restriction to the local coordinates $(x,p)$.
\end{proof}

The algebra containing these operators can be described with a Weyl Algebra defined as follows:
\begin{definition}\label{def::WeylAlgebra}
   Let $\mathbf{A}$ denote the Weyl algebra over $\C$ which is generated by $ z , \partial$ subject to the relations
 \begin{equation}
     \begin{aligned}
         z_i z_j - z_j z_i &= \partial_i \partial_j - \partial_j \partial_i = 0,
        \\ \partial_i z_j - z_j \partial_i &= \begin{cases}
             0 \text{ if } i\neq j,\\
             1 \text{ if } i=j
         \end{cases}
     \end{aligned}
 \end{equation}
 for each $i,j = 0,1,\dots,2n$.
 \end{definition}
Note that the action of an $a$ in the Weyl algebra $\mathbf{A}$ on a polynomial in $F$ is not guaranteed to return another polynomial in $F$.
Taking this fact into account, we introduce the following subsets of $\mathbf{A}$.
\begin{definition}\label{def::ScriptF}
    Let $\hat m$ be as above.
    For $m \in \Z$ define $\mathscr{F}_m \subset \mathbf{A}$  by
    \begin{equation*}
        \mathscr{F}_m \coloneqq \{w \in \mathbf{A}  \mid \hat m w - w\hat m = m w\}.
    \end{equation*}
    Additionally, define $\mathscr{F} \coloneqq \cup_m \mathscr{F}_m$.
\end{definition}
Elements of $\mathscr{F}_m$ define maps from $F_{m'} \to  F_{m'+m}$ for $m' \in \N$.
Additionally, the annihilator of the mapping
\begin{equation}
    \mathscr{F} \times F \to F/J
\end{equation}
which through theorem \ref{thm::Griffiths} describes $\mathscr{H}(X)$, clearly contains the intersection of the right ideal $\mathbf{J} = \left\langle (\hat m -1)\frac{\partial f}{\partial z_{0}} -  \partial_0, \dots , (\hat m - 1)\frac{\partial f}{\partial z_{2n}}- \partial_{2n}\right\rangle$ of $\mathbf{A}$ with $\mathscr{F}$.
A convenient reduction method comes from a localization of $\mathbf{A}$ which allows for `division' of non-zero polynomials in $\hat m$.

Let $M$ be the set of nonzero polynomials in $\hat m$.
The set is multiplicatively closed and the localization $M^{-1}\mathbf{A}$ is well defined.
Let $q_1(\hat m )^{-1},q_2(\hat m)^{-1} \in M$.
Then for any $\alpha_1,\beta_1,\alpha_2,\beta_2 \in \N^{2n+1}_0$ the product of two elements in $M^{-1}\mathbf{A}$ is determined by
\begin{equation}
    (q_1(\hat m ) z^{\alpha_1} \partial^{\beta_1}) \cdot (q_2(\hat m ) z^{\alpha_2} \partial^{\beta_2} ) =  q_1(\hat m)  q_2\left( \hat m + \frac{|\alpha_1| - |\beta_1| }{\deg f} \right)  z^{\alpha_1} \partial^{\beta_1}z^{\alpha_2} \partial^{\beta_2}
\end{equation}
and follows from the commutator $[\hat m,z^{\alpha_1} \partial^{\beta_1}]$.
Extending the ideal $\mathbf{J} $ to $M^{-1}\mathbf{A}$ gives a right ideal  $\mathbf{J}_M$.
The choice of generators
\begin{equation}
    \mathbf{J}_M =\left\langle \frac{\partial f}{\partial z_{0}} - \frac{1}{\hat m} \partial_0, \dots , \frac{\partial f}{\partial z_{2n}}- \frac{1}{\hat m}\partial_{2n}\right\rangle
\end{equation}
found by multiplying the old generators by $1/(\hat m  - 1/ \deg f )$ from the right, allows for a convenient reduction algorithm.
By ignoring the cancellations between the denominators and numerators it follows that the quotient $M^{-1}\mathbf{A} / \mathbf{J}_M $ is contained by $M \otimes (\C[z]/J) \otimes \C[\partial]$.
That is, if an element of $M^{-1}\mathbf{A}$ is in the (right) ideal $J$ defined in equation \ref{eq::JacobianIdeal}, then the degree of the factor in $\C[z]$ can be reduced.
This provides the basis of the method introduced here.
Additionally, because $M^{-1}\mathbf{A} / \mathbf{J}_M $ is a right $\mathbf{A}$-module there is a representation of the form
\begin{equation}
    \mathbf{A} \to \End(M^{-1}\mathbf{A} / \mathbf{J}_M ).
\end{equation}
We calculate the image of this map in terms of $M \otimes (\C[z]/J) \otimes \C[\partial]$.
This leads to a ring homomorphism between $\mathbf{A}$ and the ring $\mathcal{R}_{n,d}[t]$ defined in \ref{def::RingOfMatrixFamilies}.

For an element $g\in \mathbf{A}$, the representation can be determined as follows.
In the case $X$ smooth, $\C[z]/J$ is a finite dimensional vector space.
Let $d$ be the dimension of the vector space $\C[z]/J$.
Fix a basis $\{ e_{i} \}_{i=1}^{d}$ for $\C[z]/J$.
The strategy is to reduce modulo $\mathbf{J}_M$ the product of a general element $(q(\hat m )^{-1} e_i \partial^{\alpha})$ of  $M^{-1}\mathbf{A} / \mathbf{J}_M$ with $g$.
From the Leibniz formula
\begin{equation}
    \partial^{\alpha} z^{\beta} = \sum_{\nu\leq \alpha} \binom{\alpha}{\nu} \frac{\beta!}{(\beta - \nu)!} z^{\beta - \nu} \partial^{\alpha - \nu}
\end{equation}
it is straightforward to compute the polynomials $c_{i,\nu,\alpha}(z) \in \C[z]$ such that that
\begin{equation}
    (q(\hat m )^{-1} e_i \partial^{\alpha}) \cdot  g = \sum_{\nu} q(\hat m )^{-1} c_{i,\nu,\alpha}(z) \partial^{\nu}
\end{equation}
for an arbitrary $\alpha$.
The polynomial terms can now be reduced modulo $\mathbf{J}_M$ by replacing left factors of $\frac{\partial f}{\partial z_i}$ with $\hat m^{-1} \partial_i$.
For each term $c_{i,\nu,\alpha}(z)$ we can compute the polynomials $w_j$ and the $r_k \in \C$ such that
\begin{equation}
    c_{i,\nu,\alpha}(z) = \sum_{j=0}^{2n} w_{j} \frac{\partial f}{\partial z_j} + \sum_{k=1}^d r_k e_k .
\end{equation}
    The following calculation gives
    \begin{equation}\label{eq::reduction}
    \begin{aligned}
        q(\hat m )^{-1}w_{i} \frac{\partial f}{\partial z_i} &=  \frac{\partial f}{\partial z_i} \cdot \left(q\left(\hat m + \frac{\deg f -1}{\deg f} \right)^{-1} w_i \right)
        \\ & = (\hat m ^{-1} \partial_i) \cdot q\left(\hat m + \frac{\deg f -1}{\deg f} \right)^{-1} w_i
        \\ & =  \left(\hat m q\left(\hat m + 1 \right)\right)^{-1} \partial_i w_i
        \end{aligned}
    \end{equation}
    mod $\mathbf{J}_M$ for each $i =0,\dots2n$.
    Therefore,
    \begin{equation}
        q(\hat m )^{-1} c_{i,\nu,\alpha}(z) = \left(\hat m q\left(\hat m + 1 \right)\right)^{-1} \sum_{j=0}^{2n} \partial_j w_j +  q(\hat m )^{-1} \sum_{k=1}^d r_k e_k
    \end{equation}
    mod $\mathbf{J}_M$.
    The removal of the factors in $J$ can be iterated.
    To describe the result of this iteration we introduce a (right) shift operator $S_{\hat m} : M \to M$ according to
    \begin{equation}
        q(\hat m )S_{\hat m} = S_{\hat m}q(\hat m +1).
    \end{equation}
    Let $  s_{\max} =  \left\lceil\frac{\deg( c_{i,\nu,\alpha}(z))}{\deg f} \right\rceil$ and let $I_\alpha \subset \N_0^{2n+1}.$
    After no more than $s_{\max}$ iterations, we arrive at an expression of the form
    \begin{equation}\label{eq::matrices}
        (q(\hat m )^{-1} e_i \partial^{\alpha}) \cdot g = \sum_{j=1}^{d}\sum_{\beta \in I_\alpha } \sum_{s=0}^{s_{\max}}  \frac{ g^{\alpha , \beta }_{ijs}  }{ q(\hat m)} \left(S_{\hat m}\frac{1}{\hat{m} }\right)^s  e_j  \partial^{\beta}
    \end{equation}
    mod $\mathbf{J}_M$, where $g^{\alpha , \beta }_{ijs}  \in \C$.
    The set $I_{\alpha}\subset \N_0^{2n+1}$ must be finite as it is determined by the number of iterations as well as the set of $\nu' \in \N^{2n+1}_0$ such that $c_{i,\nu',\alpha}(z) \neq 0$, both of which are finite.
    Calculating the coefficients $g_{ijs}^{\alpha ,\beta}$ for each $i =1,\dots, d$ shows that there are matrices $g^{\alpha, \beta} : M\otimes (\C[z]/J) \to M\otimes (\C[z]/J) $ with components given by
    \begin{equation}
        (g^{\alpha,  \beta} )_{ij} = \sum_{s=0}^{s_{\max}} g^{\alpha , \beta }_{ijs} (( S_{\hat m} \hat m )^s \otimes 1).
    \end{equation}
    This constructs a map $ \sigma : \mathbf{A} \to \mathcal{R}_{n,d}[S_{\hat m} \hat{m}^{-1}]$.

\begin{proposition}\label{prop::RepresentationOfg}
    The map $\sigma: \mathbf{A} \to \mathcal{R}_{n,d}[S_{\hat m} \hat{m}^{-1}]$ is a ring homomorphism.
\end{proposition}
\begin{proof}
    Let $g_1,g_2 \in \mathbf{A}$.
    We have that there are $g^{\alpha , \beta }_{1,ijs}$ such that
    \begin{equation}
        (q(\hat m )^{-1} e_i \partial^{\alpha}) \cdot g_1 \cdot g_2 = \sum_{j=1}^{d}\sum_{\beta \in I_\alpha } \sum_{s=0}^{s_{1,\max}}  \frac{ g^{\alpha , \beta }_{1,ijs_2} }{ q(\hat m)}\left(S_{\hat m}\frac{1}{\hat{m} }\right)^{s_1}  e_j  \partial^{\beta} g_2
    \end{equation}
    and $g^{\beta , \nu }_{2,jk s}$ such that
    \begin{equation}
        \sum_{k=1}^{d}\sum_{\nu \in I_\beta } \sum_{s=0}^{s_{2,\max}} e_j \partial^{\beta} g_2 = g^{\beta , \nu }_{2,jk s_2} \left(S_{\hat m}\frac{1}{\hat{m} }\right)^{s_2}  e_k  \partial^{\nu}
    \end{equation}
    which implies that $\sigma(g_1 g_2) = \sum_{\beta \in I_{(\cdot)} } g_1^{\cdot, \beta}g_2^{\beta ,\cdot} = \sigma(g_1) \diamond \sigma(g_2)$.
    The map as constructed clearly satisfies $\sigma(g_1 + g_2) = \sigma(g_1)+ \sigma(g_2)$ and $\sigma(1) = \delta$.
\end{proof}

\begin{definition}\label{def::ActionOfmatrices}
    Fix a basis  $\{e_i\}_{i=1}^{d}$ for $\C[z]/J$, and suppose $e_1 = 1$.
    Let $\sigma$ be given by the above calculations in the basis $\{e_i\}_{i=1}^{d}$.
    Define 
    $$ \mu_{\mathbf{A}} : \sigma(\mathbf{A}) \otimes \C[z] \to \C[z] $$
     according to the following:
    Let $\hat{u}$ be the $d$-dimensional row vector with $1$ in the first component and zeros in all other entries.
    Let $\hat e$ be the column vector with components given by the basis $\{e_i\}_{i=1}^{d}$.
    Then for $g\in \mathbf{A}$ and $q\in \C[z]$ we have
    \begin{equation}
        \mu_{\mathbf{A}} (g , q) = \sum_{\beta \in I_{\vec 0}} \left(\hat{u} g^{\vec 0 , \beta} \hat e\right) \frac{\partial^{\beta} q}{\partial z^{\beta}}
    \end{equation}
    Further define $\mu_{\mathscr{F}} : \sigma(\mathscr{F}) \otimes F \to F$ to be the restriction of $\mu_{\mathbf{A}}$ to $\sigma(\mathscr{F}) \otimes F \subset  \sigma(\mathbf{A}) \otimes \C[z]$.
\end{definition}
The well-definedness of the map $\mu_{\mathscr{F}}$ follows from the above calculation of the coefficients $g_{ijs}^{\alpha ,\beta}$.
The map $\mu_{\mathscr{F}}$ leads to the following result:
\begin{proposition}\label{prop::MappedIntegrals}
    Let $q\in F$ and let $g \in \mathscr{F}$ then for any $2n$-cycle $\Gamma $ on $\mathbb{P}^{2n} \setminus X$ we have
    \begin{equation}
        \int_\Gamma  \psi(g q) = \int_\Gamma (\psi \circ \mu_{\mathscr{F}})\left( \sigma( g) , q \right) .
    \end{equation}
\end{proposition}

\begin{proof}
    By construction there is a $j\in  \mathbf{J}_M$ such that $\mu_{\mathscr{F}}(\sigma(g),q) = g q + jq $.
    Therefore, $$(\psi\circ \mu_{\mathscr{F}})(\sigma(g),q) = \psi(g q + jq ) = \psi(g q) +\psi( jq).$$
    Noting that $\psi(jq) \in dA^{2n-1}(X)$ gives the result.
\end{proof}

The final piece of notation useful for proving theorem \ref{thm::MainResultTwo} is established by the following definition.
\begin{definition}\label{def::PhiMap}
    Let $a \in A^{2n} (X)$ and suppose $q_m \in F_m/f\cdot F_{m-1}$ are such that 
    $ a = \sum_{m} q_m \frac{\Omega}{f^m}. $
    Define $$\varphi: A^{2n}(X)  \to F $$ according to 
    \begin{equation}
       \varphi(a) = \sum_{m} q_m.
    \end{equation}
\end{definition}

\begin{remark}
    Note that  $\psi\circ \varphi  = \text{id}$.
\end{remark}

\subsection{Proof of theorem \ref{thm::MainResultTwo}}
\label{sec::ProofOfTheoremTwo}
\begin{proof}

From proposition \ref{prop::QuantumOperatorsMapped} we can define a ring homomorphism $\tau : G \to \mathscr{F}$ according to $g_i \mapsto \hat g_i$.
Composing this with the ring homomorphism $\sigma$ introduced in proposition  \ref{prop::RepresentationOfg} constructs the ring homomorphism $\lambda = \sigma \circ \tau  : G \to \mathcal{R}_{n,d}[t]$ where $d$ is given by the dimension of the vector space $\C[z]/J$.

The map $ \mu_{\mathscr{F}} : \sigma(\mathscr{F}) \otimes F \to F$ from definition \ref{def::ActionOfmatrices}, composed with the maps $\varphi:A^{2n}(X) \to F$ from definition \ref{def::PhiMap} and $\psi:F \to A^{2n}(X)$ from definition \ref{def::PsiMap}, gives the map $\mu =  (\psi \circ \mu_{\mathscr{F}}) (\text{id}\otimes \varphi )$. %: \sigma(\mathscr{F}) \otimes A^{2n}(X) \to A^{2n}(X) $.

The second of the two properties can be checked.
For any $g\in G$ and $a\in A^{2n}(X)$ proposition \ref{prop::MappedIntegrals} shows that for any cycle $\Gamma$ on $X$ we have
\begin{equation}
    \begin{aligned}
        \int_\Gamma g a & = \int_\Gamma g (\psi\circ \varphi) (a) 
        \\ &= \int_\Gamma  \psi ( \tau(g) \varphi(a)) 
        \\ &= \int_\Gamma  (\psi\circ \mu_\mathscr{F}) ( (\sigma\circ\tau)(g), \varphi(a) ) 
        \\ &= \int_\Gamma   \mu( \lambda(g), a ) .
    \end{aligned}
\end{equation}
The first property in the statement of the theorem follows from proposition 4.13 in \cite{GriffithsI} which constrains the highest degree terms in $\C[z]/J$ to be $r$.
Therefore, the definition \ref{def::ActionOfmatrices} of the map $\mu_{\mathscr{F}}$ shows that the maximum increase in pole order is $r$.

\end{proof}

\section{Remarks on the ring of matrix valued functions}
\label{sec::StructureOfTheRingOfMatrixValuedFunctions}

The ring $\mathcal{R}_{n,d}[t]$ of definition \ref{def::RingOfMatrixFamilies} appeared as the general setting for a representation of the Weyl algebra $\mathbf{A}$ of definition \ref{def::WeylAlgebra}.
In this regard, the primary object of interest is the subring $\sigma(\mathbf{A}) \subset\mathcal{R}_{n,d}[t]$ with sigma defined by \ref{prop::RepresentationOfg}.

\subsection{The Generators of the algebra}
\label{sec::AlgebraGenerators}
The subring $\sigma(\mathbf{A})$ is determined by the images of the generators $1,z_0,\dots,z_{2n}$, $\partial_{0},\dots,\partial_{2n}$ of $\mathbf{A}$.
The following remarks record several properties of the generators of $\sigma(\mathbf{A})$.
% In the following subsection, these observations are used to propose a reduction formula akin to the Griffiths-Dwork reduction.
\begin{enumerate}
    \item The multiplicative identity in $\mathcal{R}_{n,d}[t]$ is given by the matrix valued function $\delta^{(\alpha , \beta)}$ which is the identity matrix if $\alpha = \beta$ and the zero matrix otherwise.
    
    \item Let $\gamma \in \N_0^{2n+1}$. 
    The image of $\partial^{\gamma}$ can be expressed as 
    \begin{equation}
        \sigma(\partial^{\gamma}){(\alpha , \beta)} = \delta^{(\alpha + \gamma, \beta )}
    \end{equation}
    for all $\alpha , \beta \in \N^{2n+1}_0$.

    \item Let $ u_{i} \in \N^{2n+1}_0$ be given by 
    \begin{equation}
        u_{i} = (\underbrace{0,\dots,0}_{i-1},1,0,\dots,0)
    \end{equation}
    there is a subset $I_i \subset \N^{2n+1}_0$ which is contained in 
    \begin{equation}
        I_i \subseteq  \{ \nu \in \N^{2n+1}_0 :  0 \leq |\nu| \leq  r   \}
    \end{equation}
    where $r$ is as in theorem \ref{thm::MainResultTwo}, such that for any $i =0,\dots,2n$ there are matrices $Z^{\nu}_i(\alpha)$ such that image $\sigma(z_i)$ is 
    \begin{equation}
        \sigma(z_i)(\alpha, \beta) = \sum_{\nu \in \{-u_i\} \cup I_i} Z^{\nu}_i(\alpha) \delta^{\alpha +\nu ,  \beta}.
    \end{equation}
   
    \item Only $Z^{-u_i}_i(\alpha)$ has $\alpha$ dependence and it is determined by
    \begin{equation}
        \sigma(z_i)(\alpha, \beta)  = (u_i \cdot \alpha) \delta^{\alpha  - u_{i} , \beta} + \sum_{\nu \in I_i }Z_{i}^{\nu}  \delta^{\alpha+\nu , \beta}.
    \end{equation}
\end{enumerate}
The $Z^{\nu}_i$ matrices are the fundamental building blocks of all period integrals on the surface $f=0$.
An example of these matrices is listed in section \ref{sec::ExampleZmatrices} for the cubic of equation \eqref{eq::cubicf}.

\section{Connection to WKB}
\label{sec::ConnectionToWKB}

In one dimension, the WKB expansions of $\derivative{}{E}\phi$, with $\phi$ defined by equation \eqref{eq::ActionDensity}, is related to the expansion of $ \tr \delta( \hat H -E)$ through the multivariable residue formula.
The residue formula is reviewed in section \ref{sec::TheResidueFormula}.
In section \ref{sec::RelationToEWKB} the relation to the WKB series is discussed.
Section \ref{sec::AdHocReduction} provides an second application of the Griffiths-Dwork reduction to the calculation of WKB periods which uses neither of the main results.

An additional appraoch to the computation of WKB periods comes from topological recursion \cite{fujiReconstructingGKZTopological2019,Belliard:2024pae,BouchardReconstructingWKB}.

\subsection{The residue calculus}
\label{sec::TheResidueFormula}
The spaces $\mathscr{H}^{m}(X)$ defined in \ref{def::RationaldeRhamGroup}, are related to the de Rham cohomology group $H^{2n-1}_{\text{dR}}(X)$ through the residue formula.
For the purposes of this article, the residue formula is described by the following theorem which follows from statements and results which can be found in \cite{GriffithsI}. 
Additional discussions of the residue calculus can be found in \cite{Pham_2011,Hwa_Hua_Teplitz_1966}.
\begin{theorem}[Residue Formula]\label{thm::ResidueFormula}
    Let $s:\mathbb{P}^{2n} \to \C$ be regular and let $X$ be the hypersurface defined by the equation $s=0$.
    Let $\alpha$ be a $(2n-1)$-form on $\mathbb{P}^{2n}$ regular on $X$ and let $\gamma$ be any $(2n-1)$-cycle on $X$.
    Then there exists a $2n$-cycle $T\gamma$ on $\mathbb{P}^{2n} \setminus X $  such that
    \begin{equation}
       \int_{T\gamma} \alpha\wedge \frac{\d s}{s} = 2\pi i \int_\gamma \alpha\lvert_X .
    \end{equation}
\end{theorem}
The residue formula defines a map between the rational differential forms on $\mathbb{P}^{2n} \setminus X$ and the differentials on $X$.
\begin{definition}
    For $\alpha , s ,X$ as in theorem \ref{thm::ResidueFormula}, define the residue map $\text{res}: \mathscr{H}(X) \to H^{2n-1}_{\text{dR}}(X)$ according to
    \begin{equation}\label{def::ResidueMap}
        \text{res} \left( \alpha\wedge \frac{\d s}{s} \right) = \alpha\lvert_X.
    \end{equation}
\end{definition}
\begin{remark}
    The residue map is an isomorphism
\begin{equation}\label{def::ResidueIsomorphism}
    \text{res} : \mathscr{H}(X) \to H^{2n-1}_{\text{dR}}(X)
\end{equation}
and the filtration of $\mathscr{H}(X)$ by \textit{pole order} corresponds precisely with the \textit{Hodge filtration} of $H^{2n-1}_{\text{dR}}(X)$ see \cite{GriffithsI}.
\end{remark}

The residue formula is written in terms of differential forms with pole order equal to one.
However, the residue formula can be applied to a differential form with higher order poles by first reducing the pole order using the following formula:
\begin{equation}\label{eq::PoleReductionFormula}
   \d{} \frac{\eta}{s^m} = \frac{\d{ \eta } }{s^m} -m \frac{\d{s} \wedge \eta}{ s^{m+1} }.
\end{equation}

\subsection{WKB periods}
\label{sec::RelationToEWKB}
Let $\hat H$ be the quantization of a classical hamiltonian $H$.
Consider the step function $\Theta:\R \to \R$ given by $\Theta(E) = 1$ if $E \in (-\infty,0]$ and $\Theta(E) = 0$ otherwise.
In the case that $\hat H$ has a discrete real spectrum the trace $ \tr \Theta(\hat H - E )$ counts the number of eigenvalues with energy less than $E$.
That is, for $\{ \lambda_k  \}$ the set of eigenvalues of $\hat H$ we have
\begin{equation}\label{eq::CountingFunction}
    \tr \Theta(\hat H - E ) = \# \{ \lambda_k  \mid  \lambda_k \leq E  \}.
\end{equation}
The $E$-derivative of the $\tr \Theta(\hat H - E )$ reduces the left hand side to the trace of the Dirac delta distribution
\begin{equation}
    \derivative{}{E} \tr \Theta(\hat H - E ) = \tr \delta( \hat H -E).
\end{equation}
% The trace $\tr\delta(\hat H -E)$ can be identified as the $E$-derivative of the eigenvalue counting function.
Therefore, for one degree of freedom (that is, the case $n=1$), the formal expansion of theorem \ref{thm::MainResultOne} is closely related to the counting function of quantum mechanics.
The relation to the all orders in $\hbar$ extension of the Bohr-Sommerfeld condition was formalized in \cite{Colin_de_Verdiere_2005}.
To summarize:
 The Bohr-Sommerfeld condition provides a formal asymptotic series which encodes the `counting function' associated to a specific periodic trajectory determined by a `well' of the potential function.
 A new hamiltonian which is equivalent in the classically allowed region but which is bounded outside with no additional `wells' has a counting function which agrees precisely with the Bohr-Sommerfeld condition.
 
Let $\phi$ be the `action density' defined in equation \eqref{eq::ActionDensity}.
The results of \cite{Colin_de_Verdiere_2005} as well as proposition \ref{prop::SwappingQuantizationsResolvent} and theorem \ref{thm::MainResultOne}  
 show that for $H =p^2/2 +V(x)$ the WKB periods satisfy
\begin{equation}\label{eq::ApplicationToWKB}
   \int_\gamma \derivative{}{E} \phi dx = \sum_{k=0}^{\infty} \frac{\hbar^k}{2\pi i} \int_\gamma   \text{ res}\left( P_k (g_1,g_2) \frac{dx\wedge dp}{H(x,p) - E}\right) 
\end{equation}
for any cycle $\gamma$ on the Riemann surface determined by $H=E$.

\subsection{A direct application of the residue formula to the WKB periods}
\label{sec::AdHocReduction}
For a one dimensional hamiltonian of the form $H = p^2/2 +V(x)$ the `action density' $\phi$ defined in the introduction in equation \eqref{eq::ActionDensity} satisfies
\begin{equation}\label{eq::WKBseries}
    \begin{aligned}
        \derivative{}{E} \phi  &= \frac{1}{\sqrt{2} \sqrt{E-V(x)}}+\frac{i \hbar  V'(x)}{4 (E-V(x))^2}
         \\ & +\hbar ^2 \left( -\frac{3 V''(x)}{16 \sqrt{2} (E-V(x))^{5/2}}-\frac{25 V'(x)^2}{64 \sqrt{2} (E-V(x))^{7/2}} \right) + \dots
     \end{aligned}
\end{equation}
The task of reducing these into a basis of periods of the riemann surface $H=E$ has been addressed in \cite{Basar_Dunne_Unsal_2017,Fischbach_Klemm_Nega_2019,Raman_Subramanian_2020,Ito,Mironov:2009uv,Codesido_Marino_2018} and many other places.
Here an alternate approach is described.

From the structure of the terms in equation \eqref{eq::WKBseries}, it is evident that they are residues of rational differential forms. 
The following observation demonstrates how this can be utilized for an application of the Griffiths-Dwork reduction to the computation of the WKB periods.

Let $w(x)$ be a polynomial.
The residue formula \eqref{def::ResidueMap} can be applied directly to the even terms of series \eqref{eq::WKBseries} using the following formula, which follows from induction on the pole reduction formula \eqref{eq::PoleReductionFormula}.
We have that
\begin{align}\label{eq::ResiduePhi}
    \frac{w(x)\d{x}}{(\sqrt{ E-V(x) })^{2m-1}} \d{x} &=   (m-1)! \sqrt{\frac{2}{\pi }} \Gamma \left(\frac{3}{2}-m\right) \text{res } \frac{  w(x)  \d{x}\wedge \d{p}}{(H-E)^{m}}
\end{align}
mod total derivatives.
For example, modulo exact forms
\begin{equation}\label{eq::ForExample}
    [\hbar^2] \derivative{\phi}{E} \d{x} = \text{res} \left( \frac{5 V'(x)^2}{4(H-E)^4}\d x \wedge \d p  -\frac{V''(x)}{2 (H-E)^3} \d{x} \wedge \d{p} \right).
\end{equation}

Writing the terms of the series \eqref{eq::WKBseries} as rational integrals using equation \eqref{eq::ResiduePhi} allows for the Griffiths-Dwork reduction to be applied in a term-by-term manner.
We use this for a straightforward check of equation \eqref{eq::ApplicationToWKB} in section \ref{sec::CubicExample}.

\section{The integrals}
\label{sec::Integrals}
The origin of theorem \ref{thm::MainResultOne} is now addressed.
As mentioned, it follows from a formal application of the symbol calculus.
The necessary results from the symbol calculus are briefly discussed in the following section \ref{sec::SymbolCalculus} before the proof of the main result is developed.

\subsection{The symbol calculus}
\label{sec::SymbolCalculus}

Given two semiclassical pseudodifferential operators $\text{Op}_t(A)$ and $\text{Op}_t(B)$ a natural question is: \textit{What is the symbol of} $\text{Op}_t(A) \circ \text{Op}_t(B)$?
The question is answered for the standard quantization by the following theorem.
\begin{theorem}[Standard composition rule, \cite{Zworski} p.~71]\label{thm::StandardStarProduct}
    Let $A$ and $B$ be two symbols with standard quantizations given by $\Op_1( A) $ and $\Op_1 (B) $ respectively.
    The composition $\Op_1( A) \circ \Op_1( B)$ has symbol
    \begin{equation}
         (A \star B)(x,p) = \left( e^{\frac{\hbar}{ i }\sum_{j=1}^{n} \frac{\partial}{ \partial x_j} \frac{\partial}{ \partial p_j}} A(x,p')  B(x',p) \right)\bigg\lvert_{(x',p')=(x,p)}.
    \end{equation}
\end{theorem}
The `star product' gives an efficient means to access formal series expansions of various symbols.
Of particular interest is the symbol of the resolvent $(\hat H-E)^{-1}$.
A formal series for the symbol $R$ of the resolvent can be found by solving the equation
\begin{equation}
    1 =    R\star(H-E)
\end{equation}
see for instance \cite{Colin_de_Verdiere_2005,Dimassi_Sjostrand_1999}.
For a general quantization scheme the equation can be solved recursively.
However, the standard quantization provides computational simplicity which allows for a closed form expression for the formal series solution to the above equation to be found.

It is often convenient to work with quantizations other than the standard quantization \cite{Zworski}.
For instance, the work of Colin-de-Verdi{\`e}re \cite{Colin_de_Verdiere_2005} utilized the Weyl quantization.
Any two quantization schemes can be interpolated between using the following result.
\begin{theorem}[Changing Quantizations, \cite{Zworski} p.~70]\label{thm::ChangingQuantizations}
    For any symbol $a$, let $0 \leq t\leq 1$ and let $0 \leq s\leq 1$ then
    \begin{equation}
        \Op_s(a) = \Op_t \left( e^{i(t-s) \hbar \sum_{j=1}^{n} \frac{\partial}{ \partial x_j} \frac{\partial}{ \partial p_j} }a \right).
    \end{equation}
\end{theorem}

The following equation describes traces in the symbol calculus \cite{Zworski,Dimassi_Sjostrand_1999}.
Let $a(x,p,\hbar)$ be a symbol and suppose $\hat A =  \Op_s (a )$.
Then
\begin{equation}\label{eq::Traces}
    \tr  \hat A  = \frac{1}{(2\pi \hbar )^n}\int_{\R^{2n}} a(x,p,\hbar) \d{^n x} \d{^n p} .
\end{equation}

\subsection{The Integrals}
\label{sec::theIntegrals}

In \cite{Colin_de_Verdiere_2005} Colin de Verdi\'ere used the Moyal formula for the star product of Weyl symbols to generate the terms appearing in the all orders Bohr-Sommerfeld quantization rules.
We follow a similar approach to develop a formal expansion for $\tr \delta (\hat H -E)$.
The approach is centered on the symbol $R$ of the resolvent $(\hat H -E)^{-1}$.

According to \cite{Dimassi_Sjostrand_1999,Colin_de_Verdiere_2005}, in any quantization scheme $R$ is given by an expression of the form
\begin{equation}
    R (x,p,E) = \sum_{i\geq 0}\sum_{j = 0}^{m_{i}} \frac{r_{i,j}(x,p)}{(H(x,p)-E)^{j+1}} \hbar^i
\end{equation}
where the $r_{i,j}(x,p)$ are polynomials and $m_{i}$ is an integer for each $i$.
The Helffer-Sjostrand formula \cite{Dimassi_Sjostrand_1999,Colin_de_Verdiere_2005} states that the symbol $\Delta_H$ of $\delta(\hat H -E)$ is given by
\begin{equation}\label{eq::HelfferSjostrandFormula}
    \Delta_H = \sum_{i=0}^\infty  \hbar^{i} \sum_{j} \frac{1}{j!} r_{i,j}(x,p) \frac{\partial^j \delta }{\partial u^j }\lvert_{u=H(x,p)-E} .
\end{equation}
Let $\omega = \d{x_1} \wedge \d{p_1}\wedge \dots \wedge \d{x_n} \wedge \d{p_n}$ be the volume form on the phase space $\R^{2n}$.
Applying the semiclassical trace formula \eqref{eq::Traces} gives a formal series expansion for $\tr \delta (\hat H -E )$ which is
\begin{equation}
    \begin{aligned}
        \tr \delta (\hat H -E) & =   \frac{1}{(2\pi \hbar )^n}\sum_{i=0}^\infty  \sum_{j} \frac{1 }{j!} \frac{\partial^j }{\partial E^j } \int_{\R^{2n}} r_{i,j}(x,p)  \delta (H-E)  \omega \hbar^i .
    \end{aligned}
\end{equation}
The delta function can be integrated.
That is, the $(2n-1)$-form $\omega_{2n-1}$ such that $\omega_{2n-1} \wedge \d H = \omega $ and the cycle $\gamma$ which is determined by the delta function give the expression
\begin{equation}\label{eq::trdeltaI}
    \tr \delta (\hat H -E)  =   \frac{1}{(2\pi \hbar )^n}\sum_{i=0}^\infty  \sum_{j} \frac{1 }{j!} \frac{\partial^j }{\partial E^j } \int_{\gamma } r_{i,j}(x,p)   \omega_{2n-1}\lvert_{X} \hbar^i .
\end{equation}

\begin{remark}
    The cycle $\gamma$ determined by the delta function can be constructed following the prescription discussed by Pham \cite{Pham_2011}.
\end{remark}

\begin{remark}
    For applications to exact WKB, the cycle $\gamma$ is replaced by a cycle localized around a potential well or a tunnelling sector.
\end{remark}

\begin{remark}
    For computations of period integrals, it is sufficient to calculate Picard-Fuchs equations where the cycle of integration can be identified through imposition of boundary conditions \cite{Lairez_2015}.
    As Lairez points out in the context of computing general periods of rational integrals, ``in fact there is no harm in simply discarding $\gamma$.''
    % See Lairez page 2
 \end{remark}

Applying the residue formula to equation \eqref{eq::trdeltaI} gives
\begin{equation}\label{eq::trdeltaII}
    \begin{aligned}
        \tr \delta (\hat H -E) & = \frac{1}{(2\pi \hbar )^n}\frac{1}{2\pi i}  \sum_{i=0}^\infty  \sum_{j} \frac{1 }{ j! } \frac{\partial^j }{\partial E^j }\int_{T\gamma} r_{i,j}(x,p)  \frac{\omega}{ H-E } \hbar^i
        \\ & = \frac{1}{(2\pi \hbar )^n}\frac{1}{2\pi i} \int_{T\gamma} R \omega.
    \end{aligned}
\end{equation}
The conclusion is recorded in the following proposition.
\begin{proposition}\label{prop::ResolventIntegral}
    Fix a quantization scheme and let $R$ be the symbol of the resolvent $(\hat H -E)^{-1}$ in this scheme.
    There is a cycle $T\gamma$ so that as formal power series $(2\pi \hbar )^{n}\tr \delta (\hat H -E ) = \frac{1}{2\pi i }\int_{T\gamma} R\omega$.
\end{proposition}

We remark that the change of quantization theorem \eqref{thm::ChangingQuantizations} shows how the result of theorem \eqref{thm::MainResultOne} can be adapted to other quantization schemes.
A particularly nice case is captured by the following.
\begin{proposition}\label{prop::SwappingQuantizationsResolvent}
    Suppose $H(x,p)$ satisfies $\partial^2 H/\partial x_i \partial p_i =0$ for each $i \in  \{1,\dots,n\}$.
    Then for any $s\in [0,1]$ we have the following equality of formal series
    \begin{equation*}
        \tr \delta (\Op_s( H) - E) = \tr \delta (\Op_1( H) - E).
    \end{equation*}
\end{proposition}

\begin{proof}
    Let $R_s$ and $R$ be the symbols of $\Op_s(H-E)^{-1}$ and $\Op_1(H-E)^{-1}$ respectively.
    Theorem \ref{thm::ChangingQuantizations} implies that
    \begin{equation}
        \begin{aligned}
            &\Op_s(R_s)(1) = \Op_s(R_s)\circ \Op_s(H-E)
            \\ &= \Op_1\left( e^{i(1-s)\hbar \sum_{j=1}^{n} \frac{\partial}{ \partial x_j} \frac{\partial}{ \partial p_j} } R_s \right) \circ \Op_1\left( e^{i(1-s)\hbar \sum_{j=1}^{n} \frac{\partial}{ \partial x_j} \frac{\partial}{ \partial p_j} } (H-E) \right)
            \\ &=   \Op_1\left( e^{i(1-s)\hbar \sum_{j=1}^{n} \frac{\partial}{ \partial x_j} \frac{\partial}{ \partial p_j} } R_s \right) \circ \Op_1\left( H-E \right)
        \end{aligned}
    \end{equation}
    which in terms of symbols gives
    \begin{equation}
        \left( e^{i(t-s)\hbar \sum_{j=1}^{n} \frac{\partial}{ \partial x_j} \frac{\partial}{ \partial p_j} } R_s \right) \star (H-E)=1
    \end{equation}
    and it can be concluded that $e^{i(t-s)\hbar \sum_{j=1}^{n} \frac{\partial}{ \partial x_j} \frac{\partial}{ \partial p_j} } R_s = R$.
    From Stokes' theorem
    \begin{equation}
       \frac{1}{2\pi i}\int_{T\gamma}R \omega
       = \frac{1}{2\pi i}\int_{T\gamma}  e^{i(t-s)\hbar \partial_x \partial_p } R_s\omega
      = \frac{1}{2\pi i}\int_{T\gamma}  R_s\omega
    \end{equation}
    and applying proposition \ref{prop::ResolventIntegral} gives the result.
\end{proof}

\begin{remark}
    Recall that $s=1$ corresponds to the standard quantization.
    The above proposition implies that for $n=1$ and $H=\frac{1}{2}p^2 +V(x)$.
    The trace $\delta(\hat H - E)$ is equivalent in any quantization scheme.
\end{remark}

Theorem \ref{thm::MainResultOne} is now proven.
From proposition \ref{prop::ResolventIntegral} this amounts to finding the correct expression for the symbol of the resolvent in the standard quantization.

\subsection{Proof of theorem \ref{thm::MainResultOne}}
\label{sec::ProofOfTheoremOne}

\begin{proof}
    Let $t = \deg H$.
    To prove the result, we start by a direct calculation of the symbol of the resolvent $R$ in the standard quantization.
    Applying the star product in the standard quantization gives
    \begin{equation}
        \begin{aligned}
         1 &=    R\star(H-E)
         \\ &= \left( e^{\frac{\hbar}{ i }\sum_{j=1}^{n} \frac{\partial}{ \partial x_j} \frac{\partial}{ \partial p_j}} R(x,p')(H(x',p)-E)   \right)\bigg\lvert_{(x',p')=(x,p)}
         \\ &=  (H-E)R -(H-E) \sum_{k=1}^{t} \hbar^{k}g_k R.
        \end{aligned}
     \end{equation}
     From the above calculation
     \begin{equation}
        \frac{1}{H-E}= \left( 1 - \sum_{k=1}^{t} \hbar^{k}g_k\right) R
     \end{equation}
     because the sum on the right hand side contains an over all factor of $\hbar$, the equation can be solved over formal power series using the geometric series
     \begin{equation}
        \begin{aligned}
            R &= \sum_{r=0}^\infty \left(\sum_{k=1}^{t} \hbar^{k}g_k\right)^r \frac{1}{H-E} .
        \end{aligned}
     \end{equation}
     Expanding the powers of the sum gives a sum over all permutations of $\hbar g_1,\dots,\hbar^t g_t$ with $r$ total factors.
     Using definition \ref{def::SumOfNonCommutingPermutations} we have
     \begin{equation}
        \left(\sum_{k=1}^{t} \hbar^{k}g_k\right)^r = \sum_{a_1+a_2 +\dots  + a_t = r} \hbar^{ a_1 + 2a_2 +\dots  + t a_t} \{g_1^{(a_1)}g_2^{(a_2)}\dots g_t^{(a_t)} \}.
     \end{equation}
     The coefficient of $\hbar^k$ is therefore given by the sum
    \begin{equation}
    [\hbar^k] R = \sum_{a_1+2 a_2 +\dots + t a_t = k} \{g_1^{(a_1)}g_2^{(a_2)}\dots g_t^{(a_t)} \} \frac{1}{H-E}.
    \end{equation}
Applying  proposition \ref{prop::ResolventIntegral} completes the proof of theorem \ref{thm::MainResultOne}.
\end{proof}

\section{Examples}

Some examples are discussed here which illustrate theorem \ref{thm::MainResultOne} and theorem \ref{thm::MainResultTwo}.
The Harmonic oscillator in $n$ dimensions, and a cubic potential in one dimension.

\subsection{The harmonic oscillator}
\label{sec::HarmonicOscillator}

In any dimension the harmonic oscillator is trivial.
Let 
\begin{equation}
    H = \frac{1}{2}\left( p_1^2 + \dots +p_n^2 +x_1^2 +\dots + x_n^2 \right).
\end{equation}
The Jacobian ideal $J$ defined by equation \eqref{eq::JacobianIdeal} is 
\begin{equation}
    J = \langle E z_0 ,z_1 ,z_2 , \dots ,z_{2n} \rangle.
\end{equation}
The quotient $\C[z]/J$ is therefore trivial, that is, $\C[z]/J = \{ 1\}$.
The spaces $F_1$ up to $F_{n+1}$ of definition \ref{def::Fsets} are given by
\begin{equation}
    \begin{aligned}
        F_1 &= \{0\}, 
        \\ &\vdots
     \\ F_{n} & = \{0\},
    \\  F_{n+1} &= \text{ span} \{ z_0 , z_1 , z_2, \dots, z_{2n} \}.
    \end{aligned}
\end{equation}
These observations combined with Griffiths theorem \ref{thm::Griffiths} suggests that 
$$\mathscr{H}(X) \cong F/J = 0$$
which reflects the fact that there are no closed forms which are not exact on the hypersurface $X$ defined by the harmonic oscillator.

Non-zero contributions to $\tr \delta (\hat H - E) $ come from differentials, which have an additional pole at infinity.
Consider equation \eqref{eq::ClassicalPeriodInHC} with $H$ the above harmonic oscillator.
That is,
\begin{equation}
    \frac{\omega }{H - E} = \frac{\Omega}{z_0^{2n-1 }f}
\end{equation}
which shows that the classical period is not in $A^{2n}(X)$ due to the pole on $z_0 =0$, the hyperplane at infinity.
The operators $g_1$ and $g_2$ of theorem \ref{thm::MainResultOne} are given by
\begin{equation}
    \begin{aligned}
        g_1 &= -\frac{i}{H- E} \left(p_1  \frac{\partial}{\partial x_1} + \dots + p_n  \frac{\partial}{\partial x_n} \right),
    \\  g_2 &= \frac{1}{2}\frac{1}{H- E} \left(  \frac{\partial^2}{\partial x_1^2} + \dots +   \frac{\partial^2}{\partial x_n^2} \right).
    \end{aligned}
\end{equation}
The observation, implied by proposition \ref{prop::QuantumOperatorsMapped}, that the operator $g_1$ carries an overall factor of $z_0^2$, and similarly that $g_2$ carries an overall factor of $z_0^4$ suggests that $P_k \frac{\Omega}{z_0^{2n-1}f} \in A^{2n}(X)$ if $k\geq n+1$.
Therefore, the series in \ref{thm::MainResultOne} truncates at order $n+1$.

\begin{remark}
    It is important to note that the method introduced in section \ref{sec::CompleteReduction} cannot be applied to the cases $k< n+1$.
    This is precisely for the reason mentioned above: the non-zero integrals are not integrals of forms on the smooth, irreducible variety $X$.
    The extension of the Griffiths-Dwork reduction to complete intersections of codimension two, such as what was developed in \cite{PeternellThesis}, could be used extend the method introduced in section \ref{sec::CompleteReduction} to treat these cases. 
\end{remark}

\subsection{A cubic potential in one dimension}
\label{sec::CubicExample}
    
As a non-trivial example, consider the case $n=1$ and the classical hamiltonian $H = \frac{1}{2}p^2+ x^2 + x^3.$
In the homogenous coordinates  $[z_0,z_1,z_2]$ such that $x=z_1/z_0$ and $p=z_2/z_0$ the polynomial $f$ of equation \eqref{eq::fPolynomial} is given by
\begin{equation}
    f = \frac{1}{2} z_{2}^{2} z_{0} +z_{1}^{2} z_{0}+z_{1}^{3} -E z_{0}^{3}.
\end{equation}
There are only two non-zero operators $g_1,g_2$ defined in theorem \ref{thm::MainResultOne}.
Placing the operators $g_1$ and $g_2$ in the algebra $\mathscr{F}$ as using proposition \ref{prop::QuantumOperatorsMapped} gives
\begin{equation}
    \begin{aligned}
        \hat g_1 &= i z_{0}^{3} z_{2} \left( \partial_1 -\left(2 z_{1} z_{0}+3 z_{1}^{2}\right)  \hat m \right),
    \\  \hat g_2 &= \frac{1}{2} z_0^3 \left(z_0\left( \partial_1 -\left(2 z_{1} z_{0}+3 z_{1}^{2}\right)  \hat m \right)   \right)^2. 
    \end{aligned}
\end{equation}
A monomial basis for the vector space $V = \C[z_0,z_1,z_2]/\langle \partial_0 f, \partial_1 f, \partial_2 f  \rangle$ is given by 
\begin{equation}
    S = \{1,z_{2},z_{1},z_{0},z_{2}^{2},z_{1} z_{2},z_{1}^{2},z_{1} z_{2}^{2}\}
\end{equation}
therefore $d = \dim V = 8$. 
As a result the $\lambda(g_i)^{\alpha,\beta} = g^{\alpha,\beta}_i$ constructed according to \ref{prop::RepresentationOfg} are $8\times 8$ matrices.
In total we find that for a general $\alpha$ there are 210 values of $\beta$ such that the $g_1^{\alpha ,\beta }$ are not the zero matrix and 639 values of $\beta$ such that $g_2^{\alpha ,\beta }$ are non-zero.
The map $\mu_\mathscr{F}$ of definition \ref{def::ActionOfmatrices} is defined using the vectors
\begin{equation}
    \begin{aligned}
       u = \begin{pmatrix}
        1, & 0, & \dots , & 0
    \end{pmatrix}&
       & \text{ and } &&\hat e_1 = \begin{pmatrix}
        1 \\  z_{2} \\  z_{1}\\  z_{0}\\  z_{2}^{2}\\  z_{1} z_{2}\\  z_{1}^{2}\\  z_{1} z_{2}^{2}
        \end{pmatrix}.
    \end{aligned}
\end{equation}

To correspond to a differential form a monomial must have degree which is a multiple of $3$.
There are only two elements in the basis with this property $1$ and $z_{1} z_{2}^{2}$. 
This reflects the fact that $f =0 $ describes a genus one Riemann surface.
For either independent cycle $\gamma_1,\gamma_2$ on the Riemann surface, the two periods
\begin{equation}
   \begin{aligned}
    &\pi^{(1)}(E) =  \int_{\gamma_j} \frac{\Omega}{f} &\text{ and }  && \pi^{(2)}(E) = \int_{\gamma_j} \frac{z_1 z_2^2 \Omega}{f^2}
   \end{aligned}
\end{equation}
can be identified with the solutions to the following Picard-Fuchs equations which are calculated with Griffiths-Dwork method as described in example \ref{ex::PicardFuchsEquation}.
The Picard-Fuchs equations are
\begin{equation}
    \begin{aligned}
        &4 \left(27 E^{2}-4 E\right) \derivative{{}^{2} \pi^{(1)}}{E^{2}} + 8 \left(27 E-2\right) \derivative{\pi^{(1) }}{E} + 15 \pi^{(1)}=0,
   \\   &4 \left(-27 E^{2}+4 E\right) \derivative{{}^{2} \pi^{(2)}}{E^{2}} + 108 E \derivative{ \pi^{(2) }}{E} - 3 \pi^{(2)}=0.
    \end{aligned}
\end{equation}
The two linearly independent solutions of each equation correspond to the two linearly independent choices of integration cycle for either period.
Additionally, the residue formula \eqref{thm::ResidueFormula} implies  
\begin{equation}
    \begin{aligned}
       &\text{res } \frac{\Omega}{f} = \frac{\d{x} }{p} \Big\vert_{H=E} &\text{ and, }  && \text{res } \frac{z_1 z_2^2 \Omega}{f^2} = \frac{x \d{x} }{p} \Big\vert_{H=E}.
    \end{aligned}
\end{equation}

Let $\Pi(E,\hbar) =2\pi i (2\pi \hbar )^n \tr \delta(\hat H - E)$. 
For the first two odd terms in the series the relevant calculations can be written
\begin{equation}
  \begin{aligned}
    [\hbar^1]\Pi &= \int_\Gamma \psi ( \hat u g_1^{\vec 0,\vec 0} \hat e) = 0,
 \\ [\hbar^3]\Pi &= \int_{\Gamma}\psi\left(  \hat u\left( \sum_{\beta \in I_{\vec 0} } \sum_{\nu \in I_{\beta} } g_1^{\vec 0 ,\beta} g_1^{\beta , \nu} g_1^{\nu,\vec 0} + g_1^{\vec 0,\beta} g_2^{\beta, \vec 0} + g_2^{\vec 0,\beta} g_1^{\beta, \vec 0} \right) \hat e \right)  
 \\ &= 0.
  \end{aligned}
\end{equation}
The fact that the odd terms integrate to zero is a well known feature of the WKB series.
The even term $[\hbar^2]\Pi$ is given by 
\begin{equation}\label{eq::ExampleNewMethod}
    \begin{aligned}
        [\hbar^2]\Pi &= \int_\Gamma  \psi\left(  \hat u\left( \sum_{\beta \in I_{\vec 0} }g_1^{\vec 0,\beta} g_1^{\beta, \vec 0} + g_2^{\vec 0,\vec 0} \right) \hat e \right)  
\\ & = \int_\Gamma \psi \Bigg( \left(\frac{3 \left(135 E+4\right)}{4 \left(27 E-4\right)^{2} E \left(\hat m+3\right) \left( \hat m+2\right)} \right)(1)
\\ & \qquad \qquad \qquad + \left(\frac{\left(\hat m-1\right) \left(1215 E^{2}-180 E+16\right)}{4 \left( \hat m+2\right) \left( \hat m+1\right) E^{2} \left(27 E-4\right)^{2}}\right) (z_1 z_2^2) \Bigg)
\\ & =  -\frac{135 E+4}{16 \left(27 E-4\right)^{2} E} \int_\Gamma  \frac{\Omega}{f} +\frac{1215 E^{2}-180 E+16}{48 \left(27 E-4\right)^{2} E^{2}}  \int_\Gamma \frac{z_{1} z_{2}^{2} \Omega}{f^2}
\\ & = -\frac{135 E+4}{16 \left(27 E-4\right)^{2} E}\pi^{(1)}(E)  +\frac{1215 E^{2}-180 E+16}{48 \left(27 E-4\right)^{2} E^{2}}  \pi^{(2)}(E) .
    \end{aligned}
\end{equation}
On the other hand, this can be checked using the method described in section \ref{sec::AdHocReduction}.
With $V(x) = x^2 +x^3$ the right hand side of equation \eqref{eq::ForExample} gives rational differentials on $X$. 
Placing them in homogeneous coordinates gives
\begin{equation}
    [\hbar^2]\Pi = \int_\Gamma \left(\frac{5}{4} z_0^5 z_1^2 (2 z_0+3 z_1)^2 \frac{\Omega}{f^4} - z_0^5 (z_0+3 z_1)\frac{\Omega}{f^3} \right).
\end{equation}
Applying the Griffiths-Dwork reduction to this equation reduces it to
\begin{equation}
    [\hbar^2]\Pi= -\frac{135 E+4}{16 \left(27 E-4\right)^{2} E} \int_\Gamma  \frac{\Omega}{f} +\frac{1215 E^{2}-180 E+16}{48 \left(27 E-4\right)^{2} E^{2}}  \int_\Gamma \frac{z_{1} z_{2}^{2} \Omega}{f^2}
\end{equation} 
which agrees with equation \ref{eq::ExampleNewMethod}.

\subsection{The structure of the ring of matrix functions}
\label{sec::ExampleZmatrices}

With $f$ as in equation \eqref{eq::cubicf}, the subring $\sigma(\mathbf{A})$ is determined by the images of $z_0,z_1,z_2$ under the map $\sigma$.
As described in section \ref{sec::StructureOfTheRingOfMatrixValuedFunctions}, each $\sigma(z_i)$ for $i=0,1,2$ is specified by a set of matrices $\{Z_i^\nu\}_{\nu \in I_i} $ with $I_i \subset \N^{3}_0$ a finite set.
All such non-zero matrices are listed in the below.
The function $\sigma(z_0)$ is determined by the following four matrices:
    \begin{equation}
        \begin{aligned}
            Z_0^{(0,0,0)} &= \left[\begin{array}{cccccccc}
                0 & 0 & 0 & 1 & 0 & 0 & 0 & 0 
                \\
                 0 & 0 & 0 & 0 & 0 & 0 & 0 & 0 
                \\
                 0 & 0 & 0 & 0 & 0 & 0 & -\frac{3}{2} & 0 
                \\
                 0 & 0 & 0 & 0 & \frac{1}{6 E} & 0 & \frac{1}{3 E} & 0 
                \\
                 S_{\hat m}\frac{1}{\hat m} & 0 & 0 & 0 & 0 & 0 & 0 & 0 
                \\
                 0 & 0 & 0 & 0 & 0 & 0 & 0 & 0 
                \\
                 -S_{\hat m} \frac{2}{\hat m (27 E -4)} & 0 & 0 & 0 & 0 & 0 & 0 & \frac{-3}{27 E -4} 
                \\
                 0 & 0 & S_{\hat m}\frac{1}{\hat m} & 0 & 0 & 0 & 0 & 0 
                \end{array}\right]
            \end{aligned}
        \end{equation}
        \begin{equation}
            \begin{aligned}
      \\      Z_0^{(0,0,1)} &= S_{\hat m}\frac{1}{\hat m} \left[\begin{array}{cccccccc}
                0 & 0 & 0 & 0 & 0 & 0 & 0 & 0 
                \\
                1 & 0 & 0 & 0 & 0 & 0 & 0 & 0 
                \\
                 0 & 0 & 0 & 0 & 0 & 0 & 0 & 0 
                \\
                 0 & 0 & 0 & 0 & 0 & 0 & 0 & 0 
                \\
                 0 & 1 & 0 & 0 & 0 & 0 & 0 & 0 
                \\
                 0 & 0 & 1 & 0 & 0 & 0 & 0 & 0 
                \\
                 0 & 0 & 0 & 0 & 0 & 0 & 0 & 0 
                \\
                 0 & 0 & 0 & 0 & 0 & 1 & 0 & 0 
                \end{array}\right]
            \end{aligned}
        \end{equation}
        \begin{equation}
            \begin{aligned}
                Z_0^{(0,1,0)} &= S_{\hat m}\frac{1}{\hat m}\left[\begin{array}{cccccccc}
                    0 & 0 & 0 & 0 & 0 & 0 & 0 & 0 
                    \\
                     0 & 0 & 0 & 0 & 0 & 0 & 0 & 0 
                    \\
                     \frac{1}{2} & 0 & 0 & 0 & 0 & 0 & 0 & 0 
                    \\
                     0 & 0 & 0 & 0 & 0 & 0 & 0 & 0 
                    \\
                     0 & 0 & 0 & 0 & 0 & 0 & 0 & 0 
                    \\
                     0 & 0 & 0 & 0 & 0 & 0 & 0 & 0 
                    \\
                     0 & 0 & \frac{-2}{ \left(27 E-4\right)} & \frac{9 E}{ \left(27 E-4\right)} & 0 & 0 & 0 & 0 
                    \\
                     0 & 0 & 0 & 0 & 0 & 0 & 0 & 0 
                    \end{array}\right]
            \end{aligned}
        \end{equation}
        \begin{equation}
            \begin{aligned}
            Z_0^{(1,0,0)} &= S_{\hat m}\frac{1}{\hat m}\left[\begin{array}{cccccccc}
                    0 & 0 & 0 & 0 & 0 & 0 & 0 & 0 
                    \\
                     0 & 0 & 0 & 0 & 0 & 0 & 0 & 0 
                    \\
                     0 & 0 & 0 & 0 & 0 & 0 & 0 & 0 
                    \\
                     -\frac{1}{3 E} & 0 & 0 & 0 & 0 & 0 & 0 & 0 
                    \\
                     0 & 0 & 0 & 0 & 0 & 0 & 0 & 0 
                    \\
                     0 & 0 & 0 & 0 & 0 & 0 & 0 & 0 
                    \\
                     0 & 0 & \frac{6}{\left(27 E-4\right)} & 0 & 0 & 0 & 0 & 0 
                    \\
                     0 & 0 & 0 & 0 & 0 & 0 & 0 & 0 
                    \end{array}\right]
        \end{aligned}
    \end{equation}
The function $\sigma(z_1)$ is determined by the following four matrices
\begin{equation}
    Z_1^{(0,0,0)} = \left[\begin{array}{cccccccc}
        0 & 0 & 1 & 0 & 0 & 0 & 0 & 0 
        \\
         0 & 0 & 0 & 0 & 0 & 1 & 0 & 0 
        \\
         0 & 0 & 0 & 0 & 0 & 0 & 1 & 0 
        \\
         0 & 0 & 0 & 0 & 0 & 0 & -\frac{3}{2} & 0 
        \\
         0 & 0 & 0 & 0 & 0 & 0 & 0 & 1 
        \\
         0 & 0 & 0 & 0 & 0 & 0 & 0 & 0 
        \\
        S_{\hat m}\frac{9 E}{\hat m\left(27 E-4\right)} & 0 & 0 & 0 & 0 & 0 & 0 & \frac{2}{27 E-4} 
        \\
         0 & 0 & - S_{\hat m}\frac{2}{3\hat m} & 0 & 0 & 0 & 0 & 0 
        \end{array}\right]
\end{equation}
\begin{equation}
    Z_1^{(0,0,1)} =- S_{\hat m}\frac{2}{3\hat m}\left[\begin{array}{cccccccc}
        0 & 0 & 0 & 0 & 0 & 0 & 0 & 0 
        \\
         0 & 0 & 0 & 0 & 0 & 0 & 0 & 0 
        \\
         0 & 0 & 0 & 0 & 0 & 0 & 0 & 0 
        \\
         0 & 0 & 0 & 0 & 0 & 0 & 0 & 0 
        \\
         0 & 0 & 0 & 0 & 0 & 0 & 0 & 0 
        \\
         0 & 0 & 1 & 0 & 0 & 0 & 0 & 0 
        \\
         0 & 0 & 0 & 0 & 0 & 0 & 0 & 0 
        \\
         0 & 0 & 0 & 0 & 0 & 1 & 0 & 0 
        \end{array}\right]
\end{equation}
\begin{equation}
    Z_1^{(0,1,0)} = S_{\hat m}\frac{1}{\hat m} \left[\begin{array}{cccccccc}
        0 & 0 & 0 & 0 & 0 & 0 & 0 & 0 
        \\
         0 & 0 & 0 & 0 & 0 & 0 & 0 & 0 
        \\
         0 & 0 & 0 & 0 & 0 & 0 & 0 & 0 
        \\
         \frac{1}{2 } & 0 & 0 & 0 & 0 & 0 & 0 & 0 
        \\
         0 & 0 & 0 & 0 & 0 & 0 & 0 & 0 
        \\
         0 & \frac{1}{3 } & 0 & 0 & 0 & 0 & 0 & 0 
        \\
         0 & 0 & \frac{9 E}{ \left(27 E-4\right)} & -\frac{6 E}{ \left(27 E-4\right)} & 0 & 0 & 0 & 0 
        \\
         0 & 0 & 0 & 0 & \frac{1}{3 } & 0 & 0 & 0 
        \end{array}\right]
\end{equation}
\begin{equation}
    Z_1^{(1,0,0)} = -S_{\hat m}\frac{4}{\hat m \left(27 E-4\right)}\left[\begin{array}{cccccccc}
        0 & 0 & 0 & 0 & 0 & 0 & 0 & 0 
        \\
         0 & 0 & 0 & 0 & 0 & 0 & 0 & 0 
        \\
         0 & 0 & 0 & 0 & 0 & 0 & 0 & 0 
        \\
         0 & 0 & 0 & 0 & 0 & 0 & 0 & 0 
        \\
         0 & 0 & 0 & 0 & 0 & 0 & 0 & 0 
        \\
         0 & 0 & 0 & 0 & 0 & 0 & 0 & 0 
        \\
         0 & 0 & 1 & 0 & 0 & 0 & 0 & 0 
        \\
         0 & 0 & 0 & 0 & 0 & 0 & 0 & 0 
        \end{array}\right]
\end{equation}
The function $\sigma(z_2)$ is determined by the following five matrices
\begin{equation}
    Z_2^{(0,0,0)} = \left[\begin{array}{cccccccc}
        0 & 1 & 0 & 0 & 0 & 0 & 0 & 0 
        \\
         0 & 0 & 0 & 0 & 1 & 0 & 0 & 0 
        \\
         0 & 0 & 0 & 0 & 0 & 1 & 0 & 0 
        \\
         0 & 0 & 0 & 0 & 0 & 0 & 0 & 0 
        \\
         0 & 0 & 0 & 0 & 0 & 0 & 0 & 0 
        \\
         0 & 0 & 0 & 0 & 0 & 0 & 0 & 1 
        \\
         0 & 0 & 0 & 0 & 0 & 0 & 0 & 0 
        \\
         0 & - S_{\hat m} \frac{2}{3 \hat m} & 0 & 0 & 0 & 0 & 0 & 0 
        \end{array}\right]
\end{equation}
\begin{equation}
    Z_2^{(0,0,1)} = S_{\hat m}\frac{1}{\hat m} \left[\begin{array}{cccccccc}
        0 & 0 & 0 & 0 & 0 & 0 & 0 & 0 
        \\
         0 & 0 & 0 & 0 & 0 & 0 & 0 & 0 
        \\
         0 & 0 & 0 & 0 & 0 & 0 & 0 & 0 
        \\
         1 & 0 & 0 & 0 & 0 & 0 & 0 & 0 
        \\
         0 & 0 & \frac{4}{3} & 6 E & 0 & 0 & 0 & 0 
        \\
         0 & 0 & 0 & 0 & 0 & 0 & 0 & 0 
        \\
         0 & 0 & -\frac{2}{3} & 0 & 0 & 0 & 0 & 0 
        \\
         0 & 0 & 0 & 0 & 0 & 0 & \frac{-27 E+4}{3} & 0 
        \end{array}\right]
\end{equation}
\begin{equation}
    Z_2^{(0,1,0)} = S_{\hat m}\frac{1}{3 \hat m} \left[\begin{array}{cccccccc}
        0 & 0 & 0 & 0 & 0 & 0 & 0 & 0 
        \\
         0 & 0 & 0 & 0 & 0 & 0 & 0 & 0 
        \\
         0 & 0 & 0 & 0 & 0 & 0 & 0 & 0 
        \\
         0 & 0 & 0 & 0 & 0 & 0 & 0 & 0 
        \\
         0 & -2 & 0 & 0 & 0 & 0 & 0 & 0 
        \\
         0 & 0 & 0 & 0 & 0 & 0 & 0 & 0 
        \\
         0 & 1 & 0 & 0 & 0 & 0 & 0 & 0 
        \\
         0 & 0 & 0 & 0 & 0 & -2 & 0 & 0 
        \end{array}\right]
\end{equation}
\begin{equation}
    Z_2^{(1,0,0)} =2 S_{\hat m}\frac{1}{\hat m} \left[\begin{array}{cccccccc}
        0 & 0 & 0 & 0 & 0 & 0 & 0 & 0 
        \\
         0 & 0 & 0 & 0 & 0 & 0 & 0 & 0 
        \\
         0 & 0 & 0 & 0 & 0 & 0 & 0 & 0 
        \\
         0 & 0 & 0 & 0 & 0 & 0 & 0 & 0 
        \\
         0 & 1 & 0 & 0 & 0 & 0 & 0 & 0 
        \\
         0 & 0 & 0 & 0 & 0 & 0 & 0 & 0 
        \\
         0 & 0 & 0 & 0 & 0 & 0 & 0 & 0 
        \\
         0 & 0 & 0 & 0 & 0 & 1 & 0 & 0 
        \end{array}\right]
\end{equation}

\begin{equation}
    Z_2^{(0,1,1)} = \left(S_{\hat m} \frac{1}{\hat m} \right)^2 \left[\begin{array}{cccccccc}
        0 & 0 & 0 & 0 & 0 & 0 & 0 & 0 
        \\
         0 & 0 & 0 & 0 & 0 & 0 & 0 & 0 
        \\
         0 & 0 & 0 & 0 & 0 & 0 & 0 & 0 
        \\
         0 & 0 & 0 & 0 & 0 & 0 & 0 & 0 
        \\
         0 & 0 & 0 & 0 & 0 & 0 & 0 & 0 
        \\
         0 & 0 & 0 & 0 & 0 & 0 & 0 & 0 
        \\
         0 & 0 & 0 & 0 & 0 & 0 & 0 & 0 
        \\
         3E & 0 & 0 & 0 & 0 & 0 & 0 & 0 
        \end{array}\right]
\end{equation}

\section{Statements and Declarations}

\paragraph*{Competing Interests:} 
The author states that there is no conflict of interest.

\paragraph*{Data Availability:} 
I confirm that there is no data associated with this manuscript.

\paragraph*{Acknowledgments:}
This work is supported in part by the U.S. Department of Energy, Office of High Energy Physics, Award DE-SC0010339.

 \bibliography{bibliography}

@article{AIHPA_1983__39_3_211_0,
  author    = {Voros, Andr{\'e}},
  title     = {{T}he return of the quartic oscillator. {T}he complex {WKB} method},
  journal   = {Annales de l'Institut Henri Poincar\'e, physique th\'eorique},
  pages     = {211--338},
  publisher = {Gauthier-Villars},
  volume    = {39},
  number    = {3},
  year      = {1983},
  zbl       = {0526.34046},
  mrnumber  = {729194},
  language  = {en},
  url       = {http://www.numdam.org/item/AIHPA_1983__39_3_211_0/}
}

@article{AIHPA_1999__71_1_1_0,
  author    = {Delabaere, \'Eric and Pham, Fr\'ed\'eric},
  title     = {Resurgent methods in semi-classical asymptotics},
  journal   = {Annales de l'Institut Henri Poincar\'e, physique th\'eorique},
  pages     = {1--94},
  publisher = {Gauthier-Villars},
  volume    = {71},
  number    = {1},
  year      = {1999},
  zbl       = {0977.34053},
  mrnumber  = {1704654},
  language  = {en},
  url       = {http://www.numdam.org/item/AIHPA_1999__71_1_1_0/}
}

@article{Alekseev:2024pqg,
  author        = {Alekseev, Anton and Neitzke, Andrew and Xu, Xiaomeng and Zhou, Yan},
  title         = {{WKB Asymptotics of Stokes Matrices, Spectral Curves and Rhombus Inequalities}},
  eprint        = {2403.17906},
  archiveprefix = {arXiv},
  primaryclass  = {math-ph},
  doi           = {10.1007/s00220-024-05133-0},
  journal       = {Commun. Math. Phys.},
  volume        = {405},
  number        = {11},
  pages         = {269},
  year          = {2024}
}

@article{Aniceto:2018bis,
  author        = {Aniceto, In{\^e}s and Ba\c{s}ar, Gokce and Schiappa, Ricardo},
  title         = {{A Primer on Resurgent Transseries and Their Asymptotics}},
  eprint        = {1802.10441},
  archiveprefix = {arXiv},
  primaryclass  = {hep-th},
  reportnumber  = {NSF-ITP-17-153},
  doi           = {10.1016/j.physrep.2019.02.003},
  journal       = {Phys. Rept.},
  volume        = {809},
  pages         = {1--135},
  year          = {2019}
}

@article{Ashok:2016yxz,
  author        = {Ashok, Sujay K. and Jatkar, Dileep P. and John, Renjan R. and Raman, Madhusudhan and Troost, Jan},
  title         = {{Exact {WKB} analysis of $ \mathcal{N} $ = 2 gauge theories}},
  eprint        = {1604.05520},
  archiveprefix = {arXiv},
  primaryclass  = {hep-th},
  doi           = {10.1007/JHEP07(2016)115},
  journal       = { JHEP },
  volume        = {07},
  pages         = {115},
  year          = {2016}
}

@article{Basar_Dunne_Unsal_2017,
  title        = {{Q}uantum {G}eometry of {R}esurgent {P}erturbative/{N}onperturbative {R}elations},
  issn         = {1029-8479},
  doi          = {10.1007/JHEP05(2017)087},
  abstractnote = {For a wide variety of quantum potentials, including the textbook `instanton’ examples of the periodic cosine and symmetric double-well potentials, the perturbative data coming from fluctuations about the vacuum saddle encodes all non-perturbative data in all higher non-perturbative sectors. Here we unify these examples in geometric terms, arguing that the all-orders quantum action determines the all-orders quantum dual action for quantum spectral problems associated with a classical genus one elliptic curve. Furthermore, for a special class of genus one potentials this relation is particularly simple: this class includes the cubic oscillator, symmetric double-well, symmetric degenerate triple-well, and periodic cosine potential. These are related to the Chebyshev potentials, which are in turn related to certain ${mathcal N}=2$ supersymmetric quantum field theories, to mirror maps for hypersurfaces in projective spaces, and also to topological $c=3$ Landau-Ginzburg models and `special geometry’. These systems inherit a natural modular structure corresponding to Ramanujan’s theory of elliptic functions in alternative bases, which is especially important for the quantization. Insights from supersymmetric quantum field theory suggest similar structures for more complicated potentials, corresponding to higher genus. Our approach is very elementary, using basic classical geometry combined with all-orders WKB.},
  note         = {arXiv: 1701.06572},
  number       = {5},
  journal      = {JHEP},
  author       = {Ba\c{s}ar, Gokce and Dunne, Gerald V. and \"Unsal, Mithat},
  year         = {2017},
  month        = {May},
  pages        = {87}
}

@article{Bazhanov:2003ni,
  author        = {Bazhanov, V. V. and Lukyanov, Sergei L. and Zamolodchikov, Alexander B.},
  title         = {{Higher level eigenvalues of {Q} operators and {S}chr\"odinger equation}},
  eprint        = {hep-th/0307108},
  archiveprefix = {arXiv},
  reportnumber  = {RUNHETC-2003-22},
  doi           = {10.4310/ATMP.2003.v7.n4.a4},
  journal       = {Adv. Theor. Math. Phys.},
  volume        = {7},
  number        = {4},
  pages         = {711--725},
  year          = {2003}
}

@article{Belliard:2024pae,
  author        = {Belliard, Rapha\"el and Bouchard, Vincent and Kramer, Reinier and Nelson, Tanner},
  title         = {{Highest weight vectors, shifted topological recursion and quantum curves}},
  eprint        = {2412.09120},
  archiveprefix = {arXiv},
  primaryclass  = {math-ph},
  month         = {Dec},
  year          = {2024}
}

@article{BouchardReconstructingWKB,
  author    = {Bouchard, Vincent and Eynard, Bertrand},
  title     = {Reconstructing {WKB} from topological recursion},
  journal   = {Journal de l{\textquoteright}\'Ecole polytechnique - Math\'ematiques},
  pages     = {845--908},
  publisher = {\'Ecole polytechnique},
  volume    = {4},
  year      = {2017},
  doi       = {10.5802/jep.58},
  mrnumber  = {3694097},
  language  = {en},
  url       = {http://www.numdam.org/articles/10.5802/jep.58/}
}

@article{Bucciotti:2023trp,
  author        = {Bucciotti, Bruno and Reis, Tomas and Serone, Marco},
  title         = {{An anharmonic alliance: exact WKB meets EPT}},
  eprint        = {2309.02505},
  archiveprefix = {arXiv},
  primaryclass  = {hep-th},
  doi           = {10.1007/JHEP11(2023)124},
  journal       = {JHEP},
  volume        = {11},
  pages         = {124},
  year          = {2023},
  month         = {Nov}
}

@article{Cargo_Gracia-Saz_Littlejohn_Reinsch_De_Rios_2004,
  author     = {Matthew Cargo and Alfonso Gracia-Saz and R G Littlejohn and M W Reinsch and P de M Rios},
  doi        = {10.1088/0305-4470/38/9/010},
  journal    = { J. Phys. A},
  month      = {feb},
  number     = {9},
  pages      = {1977},
  title      = {Quantum normal forms, {M}oyal star product and {B}ohr--{S}ommerfeld approximation},
  url        = {https://dx.doi.org/10.1088/0305-4470/38/9/010},
  volume     = {38},
  year       = {2005}
}

@article{Codesido_Marino_2018,
  title        = {{H}olomorphic {A}nomaly and {Q}uantum {M}echanics},
  volume       = {51},
  issn         = {1751-8113, 1751-8121},
  doi          = {10.1088/1751-8121/aa9e77},
  abstractnote = {We show that the all-orders WKB periods of one-dimensional quantum mechanical oscillators are governed by the refined holomorphic anomaly equations of topological string theory. We analyze in detail the double-well potential and the cubic and quartic oscillators, and we calculate the WKB expansion of their quantum free energies by using the direct integration of the anomaly equations. We reproduce in this way all known results about the quantum periods of these models, which we express in terms of modular forms on the WKB curve. As an application of our results, we study the large order behavior of the WKB expansion in the case of the double well, which displays the double factorial growth typical of string theory.},
  note         = {arXiv: 1612.07687},
  number       = {5},
  journal      = {J. Phys. A},
  author       = {Codesido, Santiago and Mari\~no, Marcos},
  year         = {2018},
  month        = {Feb},
  pages        = {055402}
}

@article{Colin_de_Verdiere_2005,
  title        = {Bohr-{S}ommerfeld {R}ules to {A}ll {O}rders},
  volume       = {6},
  issn         = {1424-0661},
  doi          = {10.1007/s00023-005-0230-z},
  abstractnote = {Communicated by Bernard Helffer},
  number       = {5},
  journal      = {Ann. Henri Poincar{\'e}},
  author       = {Colin de Verdi{\`e}re, Yves},
  year         = {2005},
  month        = oct,
  pages        = {925–936}
}

@article{Delabaere1997ExactSE,
  title   = {Exact semiclassical expansions for one-dimensional quantum oscillators},
  author  = {{\'E}ric Delabaere and Herv{\'e} Dillinger and F. Pham},
  journal = {J. Math. Phys.},
  year    = {1997},
  volume  = {38},
  pages   = {6126-6184}
}

@article{delaCruz:2024xit,
  author        = {de la Cruz, Leonardo and Vanhove, Pierre},
  title         = {{Algorithm for differential equations for Feynman integrals in general dimensions}},
  eprint        = {2401.09908},
  archiveprefix = {arXiv},
  primaryclass  = {hep-th},
  reportnumber  = {IPHT-t23/094, LAPTH-003/24},
  doi           = {10.1007/s11005-024-01832-w},
  journal       = {Lett. in Math. Phys.},
  volume        = {114},
  number        = {3},
  pages         = {89},
  year          = {2024}
}

@book{Dimassi_Sjostrand_1999,
  address    = {Cambridge},
  series     = {London Mathematical Society Lecture Note Series},
  title      = {Spectral {A}symptotics in the {S}emi-{C}lassical {L}imit},
  publisher  = {Cambridge University Press},
  author     = {Dimassi, M. and Sj{\"o}strand, J.},
  date       = {1999},
  doi        = {https://doi.org/10.1017/CBO9780511662195},
  collection = {London Mathematical Society Lecture Note Series}
}

@article{dimca2006generalizationgriffithstheoremrational,
  title         = {A generalization of {G}riffiths' theorem on rational integrals I},
  author        = {Alexandru Dimca and Morihiko Saito},
  journal       = {Duke Mathematical Journal},
  eprint        = {math/0501253},
  year          = {2006},
  number        = {135},
  pages         = {303 -- 326},
  archiveprefix = {arXiv},
  primaryclass  = {math.AG},
  url           = {https://arxiv.org/abs/math/0501253}
}

@article{Doran:2023yzu,
  author        = {Doran, Charles F. and Harder, Andrew and Vanhove, Pierre and Pichon-Pharabod, Eric},
  title         = {{Motivic Geometry of two-{L}oop {F}eynman {I}ntegrals}},
  eprint        = {2302.14840},
  archiveprefix = {arXiv},
  primaryclass  = {math.AG},
  reportnumber  = {IPhT-T2022/64},
  doi           = {10.1093/qmath/haae015},
  journal       = {Quart. J. Math. Oxford Ser.},
  volume        = {75},
  number        = {3},
  pages         = {901--967},
  year          = {2024}
}

@article{Dunham_1932,
  title   = {{T}he {W}entzel-{B}rillouin-{K}ramers {M}ethod of {S}olving the {W}ave {E}quation},
  volume  = {41},
  issn    = {0031-899X},
  doi     = {10.1103/PhysRev.41.713},
  number  = {6},
  journal = {Phys. Rev.},
  author  = {Dunham, J. L.},
  year    = {1932},
  month   = {Sep},
  pages   = {713–720}
}

@article{Dwork1962,
  author    = {Dwork, Bernard},
  journal   = {Publ. Math. IH{\'E}S},
  language  = {eng},
  pages     = {5-68},
  publisher = {Institut des Hautes \'Etudes Scientifiques},
  title     = {On the {Z}eta {F}unction of a {H}ypersurface: {I}},
  url       = {http://eudml.org/doc/103828},
  volume    = {12},
  year      = {1962}
}

@article{Dwork1964,
  issn      = {0003486X, 19398980},
  url       = {http://www.jstor.org/stable/1970392},
  author    = {Bernard Dwork},
  journal   = {Ann. Math.},
  number    = {2},
  pages     = {227--299},
  publisher = {[Annals of Mathematics, Trustees of Princeton University on Behalf of the Annals of Mathematics, Mathematics Department, Princeton University]},
  title     = {On the {Z}eta {F}unction of a {H}ypersurface: {II}},
  urldate   = {2024-10-15},
  volume    = {80},
  year      = {1964}
}

@article{Fischbach_Klemm_Nega_2019,
  title        = {{WKB} method and quantum periods beyond genus one},
  volume       = {52},
  issn         = {1751-8113, 1751-8121},
  doi          = {10.1088/1751-8121/aae8b0},
  abstractnote = {We extend topological string methods in order to perform WKB approximations for quantum mechanical problems with higher order potentials efficiently. This requires techniques for the evaluation of the relevant quantum periods for Riemann surfaces beyond genus one. The basis of these quantum periods is fixed using the leading behaviour of the classical periods. The full expansion of the quantum periods is obtained using a system of Picard-Fuchs like operators for a sequence of integrals of meromorphic forms of the second kind. Discrete automorphisms of simple higher order potentials allow to view the corresponding higher genus curves as covering of a genus one curve. In this case the quantum periods can be alternatively obtained using the holomorphic anomaly solved in the holomorphic limit within the ring of quasi modular forms of a congruent subgroup of SL$(2,mathbb{Z})$ as we check for a symmetric sextic potential.},
  note         = {arXiv: 1803.11222},
  number       = {7},
  journal      = {J. Phys A.},
  author       = {Fischbach, Fabian and Klemm, Albrecht and Nega, Christoph},
  year         = {2019},
  month        = {Feb},
  pages        = {075402}
}

@article{fujiReconstructingGKZTopological2019,
  title    = {Reconstructing {{GKZ}} via {{Topological Recursion}}},
  author   = {Fuji, Hiroyuki and Iwaki, Kohei and Manabe, Masahide and Satake, Ikuo},
  year     = {2019},
  month    = {Nov},
  journal  = {Commun. Math. Phys.},
  volume   = {371},
  number   = {3},
  pages    = {839--920},
  issn     = {1432-0916},
  doi      = {10.1007/s00220-019-03590-6}
}

@article{Gaiotto:2009hg,
  author        = {Gaiotto, Davide and Moore, Gregory W. and Neitzke, Andrew},
  title         = {{Wall-crossing, Hitchin systems, and the WKB approximation}},
  eprint        = {0907.3987},
  archiveprefix = {arXiv},
  primaryclass  = {hep-th},
  doi           = {10.1016/j.aim.2012.09.027},
  journal       = {Adv. Math.},
  volume        = {234},
  pages         = {239--403},
  year          = {2013}
}

@article{GriffithsI,
  issn      = {0003486X},
  url       = {http://www.jstor.org/stable/1970746},
  author    = {Phillip A. Griffiths},
  journal   = {Ann. Math.},
  number    = {3},
  pages     = {460--495},
  publisher = {Annals of Mathematics},
  title     = {On the {P}eriods of {C}ertain {R}ational {I}ntegrals: {I}},
  urldate   = {2024-04-23},
  volume    = {90},
  year      = {1969}
}

@article{GriffithsII,
  issn      = {0003486X, 19398980},
  url       = {http://www.jstor.org/stable/1970747},
  author    = {Phillip A. Griffiths},
  journal   = {Ann. Math.},
  number    = {3},
  pages     = {496--541},
  publisher = {[Annals of Mathematics, Trustees of Princeton University on Behalf of the Annals of Mathematics, Mathematics Department, Princeton University]},
  title     = {{O}n the {P}eriods of {C}ertain {R}ational {I}ntegrals: {II}},
  urldate   = {2024-08-13},
  volume    = {90},
  year      = {1969}
}

@article{Gutzwiller_1971,
  title        = {Periodic Orbits and Classical Quantization Conditions},
  volume       = {12},
  issn         = {0022-2488},
  doi          = {10.1063/1.1665596},
  abstractnote = {The relation between the solutions of the time‐independent Schrödinger equation and the periodic orbits of the corresponding classical system is examined in the case where neither can be found by the separation of variables. If the quasiclassical approximation for the Green’s function is integrated over the coordinates, a response function for the system is obtained which depends only on the energy and whose singularities give the approximate eigenvalues of the energy. This response function is written as a sum over all periodic orbits where each term has a phase factor containing the action integral and the number of conjugate points, as well as an amplitude factor containing the period and the stability exponent of the orbit. In terms of the approximate density of states per unit interval of energy, each stable periodic orbit is shown to yield a series of δ functions whose locations are given by a simple quantum condition: The action integral differs from an integer multiple of h by half the stability angle times ℏ. Unstable periodic orbits give a series of broadened peaks whose half‐width equals the stability exponent times ℏ, whereas the location of the maxima is given again by a simple quantum condition. These results are applied to the anisotropic Kepler problem, i.e., an electron with an anisotropic mass tensor moving in a (spherically symmetric) Coulomb field. A class of simply closed, periodic orbits is found by a Fourier expansion method as in Hill’s theory of the moon. They are shown to yield a well‐defined set of levels, whose energy is compared with recent variational calculations of Faulkner on shallow bound states of donor impurities in semiconductors. The agreement is good for silicon, but only fair for the more anisotropic germanium.},
  number       = {3},
  journal      = {Journal of Mathematical Physics},
  author       = {Gutzwiller, Martin C.},
  year         = {1971},
  month        = mar,
  pages        = {343–358}
}

@book{Hwa_Hua_Teplitz_1966,
  series     = {Mathematical physics monograph series},
  title      = {{H}omology and {F}eynman {I}ntegrals},
  isbn       = {978-0-8053-4750-0},
  publisher  = {W. A. Benjamin},
  author     = {Hwa, R.C. and Teplitz, V.L.},
  address    = {New York},
  year       = {1966},
  collection = {Mathematical physics monograph series}
}

@article{Ito,
  author     = {Ito, Katsushi and Yang, Jingjing},
  title      = {Exact {WKB} analysis and {TBA} equations for the {S}tark
                effect},
  journal    = {Prog. Theor. Exp. Phys.},
  fjournal   = {PTEP. Progress of Theoretical and Experimental Physics},
  year       = {2024},
  number     = {1},
  pages      = {Paper No. 013A02, 30},
  issn       = {2050-3911},
  mrclass    = {81Q20 (81Q05)},
  mrnumber   = {4691840},
  mrreviewer = {Gabriel\ Alvarez},
  doi        = {10.1093/ptep/ptad154},
  url        = {https://doi.org/10.1093/ptep/ptad154}
}

@article{iwaki2014exact,
  title     = {Exact {WKB} analysis and cluster algebras},
  author    = {Iwaki, Kohei and Nakanishi, Tomoki},
  journal   = {J. Phys.},
  volume    = {47},
  number    = {47},
  pages     = {474009},
  year      = {2014},
  publisher = {IOP Publishing}
}

@article{Jivulescu_2007,
  title     = {Exact treatment of linear difference equations with noncommutative coefficients},
  volume    = {30},
  issn      = {1099-1476},
  url       = {http://dx.doi.org/10.1002/mma.933},
  doi       = {10.1002/mma.933},
  number    = {16},
  journal   = {Mathematical Methods in the Applied Sciences},
  publisher = {Wiley},
  author    = {Jivulescu, M. A. and Messina, A. and Napoli, A. and Petruccione, F.},
  year      = {2007},
  month     = {Aug},
  pages     = {2147–2153}
}

@article{Konoplya:2019hlu,
  author        = {Konoplya, R. A. and Zhidenko, A. and Zinhailo, A. F.},
  title         = {{Higher order WKB formula for quasinormal modes and grey-body factors: recipes for quick and accurate calculations}},
  eprint        = {1904.10333},
  archiveprefix = {arXiv},
  primaryclass  = {gr-qc},
  doi           = {10.1088/1361-6382/ab2e25},
  journal       = {Class. Quant. Grav.},
  volume        = {36},
  pages         = {155002},
  year          = {2019}
}

@article{Lairez_2015,
  title     = {Computing periods of rational integrals},
  volume    = {85},
  issn      = {1088-6842},
  url       = {http://dx.doi.org/10.1090/mcom/3054},
  doi       = {10.1090/mcom/3054},
  number    = {300},
  journal   = {Math. Comput.},
  publisher = {American Mathematical Society (AMS)},
  author    = {Lairez, Pierre},
  year      = {2015},
  month     = nov,
  pages     = {1719–1752}
}

@article{lairezAlgorithmsMinimalPicard2023,
  title    = {Algorithms for Minimal {{Picard}}--{{Fuchs}} Operators of {{Feynman}} Integrals},
  author   = {Lairez, Pierre and Vanhove, Pierre},
  year     = {2023},
  month    = mar,
  journal  = {Lett. Math. Phys.},
  volume   = {113},
  number   = {2},
  pages    = {37},
  issn     = {1573-0530},
  doi      = {10.1007/s11005-023-01661-3}
}

@article{Mironov:2009uv,
  author        = {Mironov, A. and Morozov, A.},
  title         = {{Nekrasov Functions and Exact {B}ohr-{S}ommerfeld Integrals}},
  eprint        = {0910.5670},
  archiveprefix = {arXiv},
  primaryclass  = {hep-th},
  reportnumber  = {FIAN-TD-21-09, ITEP-TH-51-09},
  doi           = {10.1007/JHEP04(2010)040},
  journal       = {JHEP},
  volume        = {04},
  pages         = {040},
  year          = {2010}
}

@article{Morrison:1991cd,
  author        = {Morrison, David R.},
  editor        = {Yau, Shing-Tung},
  title         = {{{P}icard-{F}uchs equations and mirror maps for hypersurfaces}},
  eprint        = {hep-th/9111025},
  archiveprefix = {arXiv},
  reportnumber  = {DUK-M-91-14},
  journal       = {AMS/IP Stud. Adv. Math.},
  volume        = {9},
  pages         = {185--199},
  year          = {1998}
}

@article{Muller-Stach_Weinzierl_Zayadeh_2014,
  title        = {Picard-{F}uchs equations for {F}eynman integrals},
  volume       = {326},
  issn         = {0010-3616, 1432-0916},
  doi          = {10.1007/s00220-013-1838-3},
  abstractnote = {We present a systematic method to derive an ordinary differential equation for any Feynman integral, where the differentiation is with respect to an external variable. The resulting differential equation is of Fuchsian type. The method can be used within fixed integer space-time dimensions as well as within dimensional regularisation. We show that finding the differential equation is equivalent to solving a linear system of equations. We observe interesting factorisation properties of the D-dimensional Picard-Fuchs operator when D is specialised to integer dimensions.},
  note         = {arXiv: 1212.4389},
  number       = {1},
  journal      = {Commun. Math. Phys.},
  author       = {M\"uller-Stach, Stefan and Weinzierl, Stefan and Zayadeh, Raphael},
  year         = {2014},
  month        = {Feb},
  pages        = {237–249}
}

@misc{nikolaev2024geometryresurgencewkbsolutions,
  title         = {{Geometry and Resurgence of WKB Solutions of Schr\"odinger Equations}},
  author        = {Nikita Nikolaev},
  year          = {2024},
  eprint        = {2410.17224},
  archiveprefix = {arXiv},
  primaryclass  = {math.DG},
  url           = {https://arxiv.org/abs/2410.17224}
}

@phdthesis{PeternellThesis,
  author    = {Peternell, Natalie},
  publisher = {Albert-Ludwigs-Universit{\"a}t Freiburg im Breisgau},
  year      = {2018},
  month     = {Jan},
  pages     = {},
  title     = {Coherent sheaves on {C}alabi-{Y}au manifolds, {P}icard-{F}uchs equations and potential functions}
}

@book{Pham_2011,
  series    = {Universitext},
  isbn      = {978-0-85729-603-0},
  doi       = {10.1007/978-0-85729-603-0_3},
  title     = {Singularities of integrals: {H}omology, hyperfunctions and microlocal analysis},
  publisher = {Springer},
  author    = {Pham, Fr\'ed\'eric},
  address   = {London},
  year      = {2011}
}

@article{PhysicsPhysiqueFizika.2.131,
  title     = {The {B}ohr-{S}ommerfeld quantization rule and the {W}eyl correspondence},
  author    = {Argyres, P. N.},
  journal   = {Phys. Phys. Fiz.},
  volume    = {2},
  issue     = {3},
  pages     = {131--139},
  numpages  = {9},
  year      = {1965},
  month     = {Nov},
  publisher = {American Phys. Society},
  doi       = {10.1103/PhysicsPhysiqueFizika.2.131},
  url       = {https://link.aps.org/doi/10.1103/PhysicsPhysiqueFizika.2.131}
}

@article{Puhlfurst:2015foi,
  author        = {Puhlf\"urst, Georg and Stieberger, Stephan},
  title         = {{Differential Equations, Associators, and Recurrences for Amplitudes}},
  eprint        = {1507.01582},
  archiveprefix = {arXiv},
  primaryclass  = {hep-th},
  reportnumber  = {MPP-2015-150},
  doi           = {10.1016/j.nuclphysb.2015.11.005},
  journal       = {Nucl. Phys. B},
  volume        = {902},
  pages         = {186--245},
  year          = {2016}
}

@article{Raman_Subramanian_2020,
  title        = {On {C}hebyshev {W}ells: {P}eriods, {D}eformations, and {R}esurgence},
  volume       = {101},
  issn         = {2470-0010, 2470-0029},
  doi          = {10.1103/PhysRevD.101.126014},
  abstractnote = {We study the geometry and mechanics (both classical and quantum) of potential wells described by squares of Chebyshev polynomials. We show that in a small neighbourhood of the locus cut out by them in the space of hyperelliptic curves, these systems exhibit low-orders/low-orders resurgence, where perturbative fluctuations about the vacuum determine perturbative fluctuations about non-perturbative saddles.},
  note         = {arXiv: 2002.01794},
  number       = {12},
  journal      = {Phys. Rev. D},
  author       = {Raman, Madhusudhan and Subramanian, P. N. Bala},
  year         = {2020},
  month        = {Jun},
  pages        = {126014}
}

@article{vanSpaendonck:2023znn,
  author        = {van Spaendonck, Alexander and Vonk, Marcel},
  title         = {{Exact instanton transseries for quantum mechanics}},
  eprint        = {2309.05700},
  archiveprefix = {arXiv},
  primaryclass  = {hep-th},
  doi           = {10.21468/SciPostPhys.16.4.103},
  journal       = {SciPost Phys.},
  volume        = {16},
  number        = {4},
  pages         = {103},
  year          = {2024}
}

@article{Zinn-Justin:2004qzw,
  author        = {Zinn-Justin, Jean and Jentschura, Ulrich D.},
  title         = {{Multi-instantons and exact results II: Specific cases, higher-order effects, and numerical calculations}},
  eprint        = {quant-ph/0501137},
  archiveprefix = {arXiv},
  doi           = {10.1016/j.aop.2004.04.003},
  journal       = {Annals Phys.},
  volume        = {313},
  pages         = {269--325},
  year          = {2004}
}

@article{Zinn-Justin:2004vcw,
  author        = {Zinn-Justin, Jean and Jentschura, Ulrich D.},
  title         = {{Multi-instantons and exact results {I}: {C}onjectures, {WKB} expansions, and instanton interactions}},
  eprint        = {quant-ph/0501136},
  archiveprefix = {arXiv},
  doi           = {10.1016/j.aop.2004.04.004},
  journal       = {Ann. Phys.},
  volume        = {313},
  pages         = {197--267},
  year          = {2004}
}

@book{Zworski,
  author    = {Zworski, Maciej},
  editor    = {Cox, David},
  series    = {Graduate Studies in Mathematics},
  volume    = {138},
  publisher = {American Mathematical Society},
  title     = {Semiclassical {A}nalysis},
  address   = {Providence},
  year      = {2012}
}
% %% if required, the content of .bbl file can be included here once bbl is generated
% %%\input sn-article.bbl

\end{document}